\documentclass[a4paper, 12 pt, twoside, reqno]{amsart}

\usepackage[english]{babel}
\usepackage[utf8]{inputenc}

\usepackage{setspace}
\usepackage[euler-digits,euler-hat-accent]{eulervm}

\usepackage{times}
\usepackage{amsmath}
\usepackage{amsfonts}
\usepackage{amssymb}
\usepackage{amsmath}
\usepackage{amsthm}
\usepackage{graphicx}
\usepackage{array}
\usepackage{color}
\usepackage{mathrsfs}
\usepackage{hyperref}
\usepackage{eucal}
\usepackage{esint} 
\usepackage{tikz}
\usepackage{upgreek}
\usepackage{enumitem}

\usepackage{mathabx}

\newcommand{\racts}{\mathbin{\rotatebox[origin=c]{180}{$\righttoleftarrow$}}}

\newcommand{\lacts}{\mathbin{\rotatebox[origin=c]{0}{$\righttoleftarrow$}}}

\allowdisplaybreaks

\theoremstyle{plain}
\newtheorem{theorem}{Theorem}[section]

\newtheorem{proposition}[theorem]{Proposition}
\newtheorem{lemma}[theorem]{Lemma}

\theoremstyle{definition}
\newtheorem{definition}[theorem]{Definition}

\theoremstyle{remark}
\newtheorem{remark}[theorem]{Remark} 
\newtheorem{example}[theorem]{Example}

\numberwithin{equation}{section}
\numberwithin{figure}{section}
\numberwithin{table}{section}

\newcommand{\R}{\mathbb{R}}
\newcommand{\N}{\mathbb{N}}
\newcommand{\C}{\mathbb{C}}                           

\newcommand{\Z}{\mathbb{Z}}

\newcommand{\s}[1]{\CMcal{#1}}
                  
\newcommand{\bb}[1]{\mathscr{#1}}
\newcommand{\rr}[1]{\mathfrak{#1}}
\newcommand{\n}[1]{\mathbb{#1}}

\newcommand{\expo}[1]{\,\mathrm{e}^{#1}\,}

\newcommand{ \ii}{\,\mathrm{i}\,}

\newcommand{\virg}[1]{\lq\lq#1\rq\rq}                \newcommand{\ie}{\textsl{i.\,e.\,}}
\newcommand{\eg}{\textsl{e.\,g.\,}}
\newcommand{\cf}{\textsl{cf}.\,}

\hypersetup{
pdftoolbar=true,        
pdfmenubar=true,        
pdffitwindow=true,     
pdfstartview=true,    
pdftitle={NCG-Landau-II},    
pdfauthor={Giuseppe De Nittis},     
breaklinks=true, 
colorlinks=true,       
linkcolor=purple,         
citecolor=teal, 
urlcolor=blue, 
bookmarksopen=true, 
filecolor=magenta,      
}

\begin{document}

\title[]{A new light on the FKMM invariant and its consequences}

\author[G. De~Nittis]{Giuseppe De Nittis}

\address[G. De~Nittis]{Facultad de Matem\'aticas \& Instituto de F\'{\i}sica,
  Pontificia Universidad Cat\'olica de Chile,
  Santiago, Chile.}
\email{gidenittis@mat.uc.cl}

\author[K. Gomi]{Kyonori Gomi}

\address[K. Gomi]{Department of Mathematics, Tokyo Institute of Technology,
2-12-1 Ookayama, Meguro-ku, Tokyo, 152-8551, Japan.}
\email{kgomi@math.titech.ac.jp}

\vspace{2mm}

\date{\today}

\begin{abstract}
``Quaternionic'' vector bundles are the objects which describe the topological phases of quantum systems subjected to an odd time-reversal symmetry (class AII). In this work we prove that the   FKMM invariant
provides the correct fundamental characteristic class for the classification of ``Quaternionic'' vector bundles in dimension less than, or equal to three (low dimension). The new insight is provided by the interpretation of the FKMM invariant
  from the viewpoint of the Bredon equivariant cohomology. This fact, along with  basic results in equivariant homotopy theory, allows us to 
  achieve the expected result.
  
\medskip


\end{abstract}

\maketitle

\tableofcontents

\section{Introduction}\label{sect:intro}

In its simplest incarnation, a \emph{Topological Quantum System} (TQS) is a continuous matrix-valued map 
\begin{equation}\label{eq:intro_tqs0}
X\;\ni\;x\;\longmapsto\; H(x)\;\in\; \rm{Mat}_N(\C)
\end{equation}
defined on a $d$-dimensional \emph{nice}\footnote{In this work we will assume that $X$ is a topological space with the homotopy type of a finite CW complex. The \emph{dimension} $d$ of $X$ is, by definition, the maximal dimension of its cells. We will say that $X$ is low dimensional if $0\leqslant d\leqslant 3$} topological space $X$.
 Although a precise definition of TQS requires some more ingredients (see \eg \cite{denittis-gomi-14,denittis-gomi-14-gen,denittis-gomi-18-I,denittis-gomi-18-II,denittis-gomi-18-III,denittis-gomi-22}), 
one can certainly
state that
the most relevant feature of these systems is the nature of the spectrum which is made by $N$ continuous \virg{energy} bands (taking into account possible degeneracies).
It is exactly this peculiar band structure, along with the structure of the related eigenspaces, which may encode  \emph{topological} information. 
In a nutshell, assume that
one can select $m<N$ bands that do not cross the other $N-m$ bands. Then, it is possible to construct a continuous projection-valued map  $X\ni x\mapsto P(x)\in \rm{Mat}_N(\C)$ such that $P(x)$ is the rank $m$ spectral projection of $H(x)$ associated to the spectral subspace selected by the $m$ energy bands at the point $x$.
Due to the classical Serre-Swan construction \cite{serre-55,swan-62} one can associate  to $x\mapsto P(x)$ a unique (up to isomorphisms) rank $m$ complex vector bundle $\bb{E}\to X$  called the \emph{spectral bundle} (see \eg \cite[Section 2]{denittis-gomi-14}). The remarkable consequence of this construction is  the following principle:

\medskip

 \emph{One can classify the topological phases of a TQS by   the  elements of the set ${\rm Vec}_{\C}^{m}(X)$ of isomorphism classes of rank  $m$ complex vector bundles over $X$}. 
 
 \medskip
 
 Therefore, the problem of the enumeration of the  topological phases of a TQS  can be  converted into the classical problem in topology of the classification of 
${\rm Vec}_{\C}^{m}(X)$. The important result due to F. P. Peterson
\cite{peterson-59}  establishes that  this classification can be achieved  by using the Chern classes which take values in the cohomology groups $H^{2k}(X,\Z)$.
In particular,  in  \emph{low dimension} 
the classification is completely specified by the first Chern class $c_1$, \ie
\begin{equation}\label{eq:intro_tqs3}
c_1\;:\;{\rm Vec}_{\C}^{m}(X)\;\stackrel{\simeq}{\longrightarrow}\;H^2\big(X,\Z\big)\;,\qquad\quad\forall\ m\in\N \quad\text{if} \quad d\leqslant 3\;.
\end{equation}

\medskip

TQS of type \eqref{eq:intro_tqs0}
 are ubiquitous in mathematical physics \cite{bohm-mostafazadeh-koizumi-niu-zwanziger-03,chruscinski-jamiolkowski-04}. 
 They can be used to model
systems subjected to  \emph{cyclic adiabatic processes} in  classical and quantum mechanics \cite{pancharatnam-56,berry-84}, or in the description of the \emph{magnetic monopole} \cite{dirac-31, yang-96}
and  the \emph{Aharonov-Bohm effect} \cite{aharonov-bohm-59}, 
or in the molecular dynamics in the context of the  \emph{Born-Oppenheimer approximation} \cite{baer-06,faure-zhilinskii-01,gat-robbins-15},
 just to mention some important example.

  \medskip
 
A very important example of  TQS comes from the Condensed Matter Physics
and concerns the dynamics of (independent) electrons in a periodic background (a crystal). In this case the Bloch-Floquet formalism \cite{ashcroft-mermin-76,kuchment-93} allows us to decompose the Schr\"odinger operator in a parametric family of operators like in \eqref{eq:intro_tqs0} labelled by the points of a torus $X=\n{T}^d$ ($d=1,2,3$),
called \emph{Brillouin zone}. In this particular case the classification of the topological phases  is completely specified by $H^2(\n{T}^d,\Z)$ due to \eqref{eq:intro_tqs3}, and the different topological phases are interpreted as the distinct \emph{quantized} values of the Hall conductance by means of the celebrated \emph{Kubo-Chern formula} \cite{thouless-kohmoto-nightingale-nijs-82,bellissard-elst-schulz-baldes-94}. The last result is the core of the theoretical explanation
of the \emph{Quantum Hall Effect} which is the prototypical example of topological insulating phases. Nowadays the study 
of topologically protected phases of \emph{topological insulators} is a mainstream topic in Condensed Matter Physics (see the  reviews  \cite{hasan-kane-10} and \cite{ando-fu-15} for a vast  overview on the subject).

\medskip

The problem of the classification of the topological phases becomes more interesting, and challenging, when the TQS is constrained by  certain \emph{(pseudo-)symmetries} like the {time-reversal symmetry} (TRS). A system like \eqref{eq:intro_tqs0} is said to be time-reversal symmetric if there is an  \emph{involution} $\tau:X\to X$ on the base space and an \emph{anti}-unitary map $\Theta$ such that
\begin{equation}\label{eq:intro_tqs3bis}
\left\{
\begin{aligned}
\Theta\;H(x)\;\Theta^*\;&=\; H\big(\tau(x)\big)\;,&\quad\qquad \forall\ x\in\ X\;\\
\Theta^2\;&=\;\epsilon\;{\bf 1}_{N}& \epsilon=\pm1\;
\end{aligned}
\right.
\end{equation}
where ${\bf 1}_{N}$ denotes the $N\times N$ identity matrix.
Let us point out that in the  definition above a crucial role is played by the pair $(X,\tau)$ which is called \emph{involutive space}\footnote{
In this case we will assume that $(X,\tau)$ has the equivariant homotopy type of a finite $\Z_2$-CW complex in the sense explained in Section \ref{sec_eq_cw}.}. Its fixed point set will be denoted by $X^\tau:=\{x\in X\;|\; \tau(x)=x\}$.

\medskip

The case $\epsilon=+1$ corresponds to an \emph{even}  TRS. In this case, the spectral 
 bundle $\bb{E}$ turns out to be equipped with an additional structure  \cite[Section 2]{denittis-gomi-14}, which makes it  a  \emph{\virg{Real} vector bundle} in the sense of Atiyah \cite{atiyah-66}. Therefore, in the presence of an even TRS the  classification problem of the topological phases 
of a TQS is reduced to the study of the set  
${\rm Vec}_{\rr{R}}^{m}(X,\tau)$ of equivalence classes of rank $m$ \virg{Real} vector bundles over the involutive space $(X,\tau)$. Also in this case,  there is a complete classification in low-dimension given by 
\begin{equation}\label{eq:intro_tqs4}
c_1^{\rr{R}}\;:\;{\rm Vec}_{\rr{R}}^{m}(X,\tau)\;\stackrel{\simeq}{\longrightarrow}\;H^2_{\Z_2}(X,\Z(1))\;,\qquad\quad\forall\ m\in\N \quad\text{if} \quad d\leqslant 3\;,
\end{equation}
which provides the generalization of \eqref{eq:intro_tqs3}.
The  cohomology appearing in \eqref{eq:intro_tqs4} is the \emph{equivariant Borel cohomology} of the involutive space $(X,\tau)$
with local coefficient system $\Z(1)$  (see Appendix \ref{subsec:borel_cohom} and references therein for more details).
The isomorphism \eqref{eq:intro_tqs4}, called Kahn's isomorphism, is induced by the \virg{Real} Chern classes $c_j^{\rr{R}}$ as defined in  \cite{kahn-59}. Its proof follows from
\cite[Proposition 1]{kahn-59}
(see also  \cite[Corollary A.5]{gomi-15}) along with  the stable condition for \virg{Real} vector bundles  \cite[Theorem 4.25]{denittis-gomi-14}.

\medskip

The case $\epsilon=-1$ describes an \emph{odd}  TRS. Also in this situation the spectral 
vector bundle $\bb{E}$ acquires an additional structure which converts $\bb{E}$ in a  \emph{\virg{Quaternionic} vector bundle} in the sense of Dupont \cite{dupont-69}.
Therefore, the topological phases of a TQS with an odd TRS are labelled  by the set ${\rm Vec}_{\rr{Q}}^{m}(X,\tau)$ 
of equivalence classes of   \virg{Quaternionic} vector bundles over  $(X,\tau)$.
The study of   systems with an odd TRS is generally more interesting, and usually harder, than the even case. Historically, the fame of these \virg{fermionic-type} systems is related with the seminal papers \cite{kane-mele-05,fu-kane-mele-95}  where the \emph{Quantum Spin Hall Effect} is interpreted as the manifestation of a non-trivial topology for TQS constrained by an odd TRS.

\medskip

Due the relevance of these systems, it would be certainly important
to have a formula for the classification of low-dimensional 
\virg{Quaternionic} vector bundles, which
generalizes the classifications 
\eqref{eq:intro_tqs3} for complex  vector bundles, or the 
classifications   \eqref{eq:intro_tqs4} for \virg{Real} vector bundles.
Such an achievement  is precisely the main result of this work. It can  be stated as
\begin{equation}\label{eq:intro_tqs43}
\kappa\;:\;{\rm Vec}_{\rr{Q}}^{m}(X,\tau)\;\stackrel{\simeq}{\longrightarrow}\;H^2_{\Z_2}(X|X^\tau,\Z(1))\;,\qquad\quad\forall\ m\in\N \quad\text{if} \quad d\leqslant 3\;,
\end{equation}
and is proved in Theorem \ref{thm:main_in_the_body}.
The cohomology group in the right-hand side of \eqref{eq:intro_tqs43} is the  second Borel equivariant cohomology group of the pair $(X, X^\tau)$ with coefficients in the local system $\Z(1)$ (see Appendix \ref{subsec:borel_cohom}). The element $\kappa$ is called  FKMM invariant and has been studied in \cite{denittis-gomi-14-gen,denittis-gomi-18-I,denittis-gomi-18-II}. By virtue of the isomorphism \eqref{eq:intro_tqs43}, and its comparison with  \eqref{eq:intro_tqs3} and \eqref{eq:intro_tqs4}, we can reformulate our main result as follow:

\medskip

\emph{The FKMM-invariant is the  fundamental characteristic class for the category of \virg{Quaternionic} vector bundles in the sense that it  completely classifies \virg{Quaternionic} vector bundles in low-dimension ($d\leqslant 3$).} 

\medskip 

Let us spend a few words about the history of the isomorphism 
\eqref{eq:intro_tqs43}. 
First of all it is worth noting that odd-rank \virg{Quaternionic} vector bundles can be defined only when $X^\tau=\emptyset$, meaning that the $\Z_2$-action induced by $\tau$ on $X$ is free.
In such a case the relative cohomology group in \eqref{eq:intro_tqs43} reduces to the ordinary cohomology group $H^2_{\Z_2}(X,\Z(1))$
and the result has been proven in \cite[Theorem 1.2 \& Theorem 1.3 (i)]{denittis-gomi-18-II}. Henceforth, let us assume  $X^\tau\neq\emptyset$, and consequently $m\in 2\N$. Initially, the isomorphism
\eqref{eq:intro_tqs43} has been established in the case of certain special involutive spaces called \emph{time-reversal} spheres and tori 
(see Example \ref{ex:TR-ts})
in \cite{denittis-gomi-14-gen}. The generalization to the case of spheres and tori with all possible involutions has been obtained in 
\cite{denittis-gomi-18-I}. The  isomorphism \eqref{eq:intro_tqs43} in the case of a general involutive space with the homotopy type of a $\Z_2$-CW complex of dimension $d\leqslant 2$ has been proven in \cite[Theorem 1.3 (ii)]{denittis-gomi-18-II}. However, in the same claim it is also stated that the map \eqref{eq:intro_tqs43}  in dimension $d=3$ is only injective. The latter claim was based on the study of a potential  counterexample that apparently   violates \eqref{eq:intro_tqs43} in dimension $d=3$ \cite[Section 5]{denittis-gomi-18-II}. It turns out that one of the key argument for the construction of the counterexample results wrong \cite[Lemma 5.3]{denittis-gomi-18-II},
and in fact the isomorphism \eqref{eq:intro_tqs43} is valid also in this case as checked in Section \ref{sec:lens}. 

\medskip

On the one hand, the new contribution of this work consists in the correct extension of  \cite[Theorem 1.3 (ii)]{denittis-gomi-18-II} to the case $d=3$. However, this is not the only benefit. In fact the technique used in this work to achieve Theorem \ref{thm:main_in_the_body}, and in turn   \eqref{eq:intro_tqs43}, is completely different from that used in 
\cite{denittis-gomi-18-II}, and sheds a new light on the nature of the FKMM invariant.

\medskip

In equivariant homotopy theory, the most basic equivariant cohomology theory is the so-called \emph{Bredon equivariant cohomology} \cite{Bre}, rather than the Borel equivariant cohomology
that appears in the definition of the FKMM invariant.
 There are, of course, some relationship between these two equivariant cohomology theories. A crucial step in our analysis consists in showing that the Borel equivariant cohomology $H^n_{\Z_2}(X| X^\tau,\Z(1))$ 
 can be indeed interpreted as  a Bredon equivariant cohomology (Theorem \ref{teo_nat_iso}). This discovery, along with the specific properties of the Bredon equivariant cohomology, allow us to derive the isomorphism \eqref{eq:intro_tqs43} along an elegant, and conceptually new path

\medskip

\noindent
{\bf 
Structure of the paper.}
In Section \ref{sec:Bredon}, we introduce the basic concepts concerning the  Bredon equivariant cohomology.  Section \ref{sec:FKMM} is devoted to the study of the  FKMM invariant from the viewpoint of the Bredon equivariant cohomology and contains the proofs of Theorem   \ref{teo_nat_iso} and Theorem   \ref{thm:main_in_the_body}. In Section \ref{sec:lens}, we 
test   our main result, \ie the isomorphism \eqref{eq:intro_tqs43}, for the 3-dimensional lens space. In particular this gives us a way to amend some errors contained in  \cite{denittis-gomi-18-II}.
  Appendix \ref{subsec:borel_cohom}, written for the benefit of the reader, provides a soft introduction to the Borel equivariant cohomology.
Finally, Appendix \ref{subsec:cohom_EilenbergMacLane} contains some technical result needed for the proof of Theorem \ref{thm:main_in_the_body}.

\medskip

\noindent
{\bf Acknowledgements.}
GD's research is supported by the grant {Fondecyt Regular - 1190204}. KG is supported by JSPS KAKENHI Grant Numbers 20K03606. 


\section{Basic facts about the Bredon  cohomology}
\label{sec:Bredon}

This section is devoted to the introduction of the \emph{Bredon (equivariant) cohomology}. The main reference is Bredon's original work \cite{Bre}, but also  the monograph \cite{May} and the paper \cite{illman-73} will be useful. 
Since we will need some concepts of category theory we will refer to the monographs \cite{freyd-64,maclane-78} as a useful reference.
Let $\n{G}$ be a topological group. By  a \emph{$\n{G}$-space} $X$ we mean a topological space $X$ together with a (left) action of $\n{G}$ on $X$ by homeomorphisms. The  expression  \emph{involutive space} will be used  as a synonym for $\n{Z}_2$-space.
In its full generality the Bredon equivariant cohomology can be defined for $\n{G}$-spaces with the action of a compact Lie group $\n{G}$. However, for the aims of this work, 
it will be enough (when necessary) to restrict our interest to the simpler case
 where $\n{G}$ is a \emph{finite} group.

\subsection{Orbit category}
Let $\n{H}$ be a closed subgroup of $\n{G}$ and denote with $\n{G}/\n{H}$ the corresponding (left) coset space.
We will write $[g]_{\n{H}} \in \n{G}/\n{H}$ to mean the element represented by $g \in \n{G}$, namely $[g]_\n{H}:=\{gh\;|\; h\in\n{H} \}$.
Therefore, by definition, $[g]_{\n{H}} = [gh]_{\n{H}}$ for any $h \in \n{H}$.
The  operation $g'\cdot[g]_{\n{H}}:=[g'g]_{\n{H}}$, with $g'\in\n{G}$, shows
that $\n{G}$ acts (on the left) on  $\n{G}/\n{H}$, and under this action 
$\n{G}/\n{H}$ turns out to be an homogeneous space.
Let $\n{H}$ and $\n{K}$ be closed subgroups of $\n{G}$.
A map $\phi:\n{G}/\n{H}\to \n{G}/\n{K}$ is called $\n{G}$-\emph{equivariant} if $\phi(g'\cdot[g]_{\n{H}})=g'\cdot \phi([g]_{\n{H}})$
for all $g,g'\in\n{G}$. The family of $\n{G}$-{equivariant} maps admits a simple description as proved in  \cite[Section I.3]{Bre}. In fact 
any $\n{G}$-{equivariant} map $\phi:\n{G}/\n{H}\to \n{G}/\n{K}$ is realized as $\phi([g]_{\n{H}}) = [ga]_{\n{K}}$, where   the element $a \in \n{G}$ meets $a^{-1} \n{H} a \subseteq \n{K}$.

\begin{definition}
Let $\n{G}$ be a topological group. The \emph{orbit category} $\mathtt{Orb}_{\n{G}}$ is defined as the category such that the
objects are the  $\n{G}$-space $\n{G}/\n{H}$, where $\n{H} \subseteq \n{G}$ are closed subgroups, and the morphisms are $\n{G}$-equivariant maps.
\end{definition}

\medskip

Let $X$ be a $\n{G}$-space and denote with $O_{\n{G}}(x):=\{g\cdot x\;|\; g\in\n{G}\}$ the \emph{orbit}  of the action through the point $x\in X$.
It turns out that $O_{\n{G}}(x):=\{g\cdot x\;|\; g\in\n{G}\}$
is isomorphic, as a topological $\n{G}$-space, to a coset space 
$\n{G}/\n{H}_x$ where $\n{H}_x:= \{g\in\n{G} \;|\; g\cdot x=x \}$ is the \emph{stabilizer group} of the point $x$.
Therefore, the orbit category of a group $\n{G}$
 gives the category of ``all kinds'' of orbits of 
$\n{G}$.

\medskip

We will write $X^{\n{H}} \subset X$ for the subspace consisting of 
all the points fixed by the $\n{H}$-action, \ie $X^{\n{H}}:=\{x\in X\;|\; hx = x\;\forall h\in \n{H}\}$. It turns out that, if there is $a \in \n{G}$ such that $a^{-1} \n{H} a \subseteq \n{K}$, then there is a map $X^{\n{K}} \to X^{\n{H}}$ given by $x \mapsto ax$. 

\medskip

\begin{example}[Orbit category of $\Z_2$]\label{ex:orb_catZ2}
In view of its relevance for this work, it will be useful to construct 
the orbit category of the cyclic group of order two
$\Z_2 = \{ \pm 1 \}$. In this case, there are only two possible subgroups of $\Z_2$, namely $\n{H}_0:=\{+1\}$ and $\n{H}_1:=\Z_2$. Therefore, 
the associated coset spaces $Z_i:=\Z_2/\n{H}_i$, with $i=0,1$, have a simple description: $Z_0=\{z_+,z_-\}$ is a two-point space (identifiable with $\Z_2$), while $Z_1=\{z\}$ is a singleton.
These two spaces are the only objects of the  orbit category of 
$\mathtt{Orb}_{\n{Z}_2}$. To complete the description of $\mathtt{Orb}_{\n{Z}_2}$ we need to describe all the possible  
$\n{Z}_2$-equivariant maps between the objects. 
In addition to the two \emph{identity} maps ${\rm Id}_i:Z_i\to Z_i$, with $i=0,1$, induced by the action of $+1$, there are two more non-trivial maps. The space $Z_0$ admits a second internal map $\phi_-:Z_0 \to Z_0$ defined by $\phi_-(z_\pm):=z_\mp$, which is induced by the action of $-1$. Then, there is the map $\psi:Z_0\to Z_1$ defined by $\psi(z_\pm):=z$, which is induced by the action of $\pm 1$, irrespectively, in view of the fact that $a^{-1}\n{H}_0 a=\n{H}_0$ for every choice $a=\pm 1$. The category $\mathtt{Orb}_{\n{Z}_2}$ can be represented by the following diagram:
$$
_{{\rm Id}_0,\phi_-}\racts\; Z_0\; \stackrel{\psi}{\longrightarrow}\; Z_1\;{{\lacts}}_{\;{\rm Id}_1}\;.
$$
The fact that there are no equivariant maps from $Z_1$ to $Z_0$ is a consequence of the fact that  $a^{-1}\n{H}_1 a=\Z_2$,  for every $a=\pm1$,
can never be a subset of $\n{H}_0$.
 \hfill $\blacktriangleleft$
\end{example}

\subsection{Equivariant CW complexes}\label{sec_eq_cw}
Let $\n{G}$ be a finite group. A $\n{G}$-equivariant CW complex, or
\emph{$\n{G}$-CW complex} for short, is a $\n{G}$-space which is homotopy equivalent to a CW-complex where the usual $n$-dimensional cells modeled by the $n$-dimensional disk $\n{D}^n$ are replaced by $\n{G}$-equivariant cells of the form $\n{G}/\n{H}\times \n{D}^n$ for some subgroup $\n{H}\subseteq \n{G}$. 
For a precise definition of $\n{G}$-CW complex we refer to \cite[Definition 1.1.1]{allday-puppe-93} or \cite[Section I.1]{Bre}
Since the notion of $\n{G}$-CW complex
 is modeled after the usual definition of CW-complex after replacing the \virg{cells} by \virg{$\n{G}$-cells} it follows that many topological and homological properties of CW-complexes have their \virg{natural} counterparts in the equivariant setting. 
 
\begin{example}[Time-reversal spheres and tori]\label{ex:TR-ts}
Special examples of  $\Z_2$-spaces are provided by the spheres 
 $\n{S}^d\subset \R^{d+1}$ (of radius 1) endowed with the \emph{time-reversal involution} $\tau:(x_0,x_1,\ldots,x_d)\mapsto (x_0,-x_1,\ldots,-x_d)$. A similar type of involution (still called $\tau$) can be defined for  tori $\n{T}^d:=\n{S}^1\times\ldots\times \n{S}^1$ just acting with the time-reversal involution on each copy of the 1-sphere. The $\Z_2$-spaces $(\n{S}^d,\tau)$ and  $(\n{T}^d,\tau)$
 are $\n{Z}_2$-CW complexes, and their structures  have been explicitly described in \cite[Section 4.5]{denittis-gomi-14}.
 \hfill $\blacktriangleleft$
\end{example}

\subsection{Coefficient systems}

Let $\mathtt{Ab}$ be  the \emph{category of abelian groups} and $\n{G}$ a finite group.
Following  \cite[Section I.4]{Bre}, we introduce the following concept:
\begin{definition}
A \textit{coefficient system} for $\n{G}$ is a contravariant functor $\bb{M} : \mathtt{Orb}_{\n{G}} \to \mathtt{Ab}$.
\end{definition}

\medskip

If $\bb{M},\bb{N} : \mathtt{Orb}_{\n{G}} \to \mathtt{Ab}$ are coefficient systems, a  morphism $F:\bb{M}\to\bb{N}$ is a natural transformation of functors. In this way the collection of coefficient systems for $\n{G}$
form an \emph{abelian category}\footnote{For the definition of abelian category we will refer to \cite[Chapter VIII]{maclane-78} or \cite{freyd-64}.}  denoted with $\mathtt{CoSy}_{\n{G}}$.

\medskip

Let us provide some  basic examples of coefficient systems.

\begin{example}\label{ex-homot}
Let $X$ be a (path connected) $\n{G}$-space which admits a fixed point $x_\ast$. For  $n\in\N$, the coefficient system $\underline{\pi}_n[X]$ is defined as 
\[
\underline{\pi}_n[X]\;:\;\n{G}/\n{H} \;\longmapsto\; \pi_n(X^{\n{H}})
\]
for every objects $\n{G}/\n{H}\in \mathtt{Orb}_{\n{G}} $,
where $\pi_n(X^{\n{H}})$ is the $n$-th homotopy group of the space $X^{\n{H}}$ computed with respect to the fixed point $x_\ast$ (which is evidently contained in each $X^{\n{H}}$). 
Let us notice that $\pi_n(X^{\n{H}})\in \mathtt{Ab}$ whenever $n\geqslant 2$. However, for the case $n=1$ the condition $\pi_1(X^{\n{H}})\in \mathtt{Ab}$  amounts to an assumption on the $\n{G}$-space $X$.
If there is a $\n{G}$-equivariant map $\phi : \n{G}/\n{H} \to \n{G}/\n{K}$ realized as $\phi([x]_\n{H}) = [xa]_\n{K}$ by an element $a \in \n{G}$ such that $a^{-1} \n{H} a \subseteq \n{K}$, then the map $f_\phi:X^\n{K} \to X^\n{H}$ given by $f_\phi(x) := ax$ induces a map
$f_{\phi,\ast}: \pi_n(X^{\n{K}})\to \pi_n(X^{\n{H}})$ between the homotopy groups. Identifying $f_{\phi,\ast}$ with the image of $\phi$ under the functor $\underline{\pi}_n[X]$ one gets that
$$
\underline{\pi}_n[X](\phi)\; :\; 
\underline{\pi}_n[X]( \n{G}/\n{K})\; \longrightarrow
\;\underline{\pi}_n[X]( \n{G}/\n{H})\;,
$$
showing that $\underline{\pi}_n[X]$ is contravariant.
 \hfill $\blacktriangleleft$
\end{example}

\begin{example}\label{ex:main_coe:sys}
Let $X$ be a $\n{G}$-CW complex. For any non-negative integer $n\in\{0\}\cup\N$, we write $X_n \subset X$ for the $n$-skeleton of $X$. Then, we have a coefficient system $\underline{C}_n[X]$ defined by 
$$
\underline{C}_n[X]\;:\;\n{G}/\n{H} \;\longmapsto\; H_n^{CW}(X^\n{H})\;:=\;H_n\big(X_n^\n{H}| X_{n-1}^\n{H}, \Z\big)\;,
$$
where on the right hand side there is the $n$-th \emph{cellular} homology group of  $X^\n{H}$, 
which is naturally isomorphic to
 the  relative singular homology (with integer coefficient) of the pair $(X_n^\n{H}, X_{n-1}^\n{H})$. 
The fact that $\underline{C}_n[X]$ is contravariant can be checked exactly as in Example \ref{ex-homot}. The connecting homomorphism in the exact sequence for the triple $(X^\n{H}_n, X^\n{H}_{n-1}, X^\n{H}_{n-2})$ provides a homomorphism
$$
\partial_n \::\; H_n^{CW}(X^\n{H})\; \longrightarrow\; 
H_{n-1}^{CW}(X^\n{H})
$$
which  is natural. Therefore it induces a natural transformation 
 $\underline{\partial}_n : \underline{C}_n[X] \to \underline{C}_{n-1}[X]$ such that $\underline{\partial}_{n-1} \circ\underline{\partial}_n=0$. In other words, $\underline{C}_\bullet[X]$ gives rise to a contravariant functor from $\mathtt{Orb}_{\n{G}}$ to the category 
 $\mathtt{ChComp}$
 of chain complexes.
 \hfill $\blacktriangleleft$
\end{example}

\begin{example}\label{ex:hh}
Let $h^\bullet = \{ h^n \}_{n \in \Z}$ be a \emph{$\n{G}$-equivariant generalized cohomology theory}. By this it is meant that $h^\bullet$ satisfies the straightforward $\n{G}$-equivariant generalization of the standard Eilenberg-Steenrod  axioms of a generalized cohomology theory \cite[Section I.3]{eilenberg-steenrod-52}. A more detailed explanation is given in \cite[Section I.2]{Bre}.
With this, we can define the coefficient system $\underline{h}^n$, with  $n \in \Z$, simply by
$$
\underline{h}^n\;:\;\n{G}/\n{H} \;\longmapsto\; h^n(\n{G}/\n{H})
$$
in view of the fact every $h^n$ is a contravariant functor from the category of (pointed) topological spaces to $\mathtt{Ab}$, by definition.
\hfill $\blacktriangleleft$
\end{example}

\begin{example}[Coefficient systems for $\Z_2$]\label{ex:coef_sysZ2}
Let us now discuss in some details the structure of the coefficient systems
for $\Z_2$ by using the notation introduced in Example \ref{ex:orb_catZ2}. In this case any coefficient system $\bb{M} : \mathtt{Orb}_{\n{Z}_2} \to \mathtt{Ab}$ is specified by two abelian groups $\bb{M}(Z_i)$, with $i=1,2$, together with two  homomorphisms (in addition to the identity homomorphisms)
$$
\bb{M}(\phi_-)\;:\; \bb{M}(Z_0)\;\longrightarrow\; \bb{M}(Z_0)\;,\qquad
\bb{M}(\psi)\;:\; \bb{M}(Z_1)\;\longrightarrow\; \bb{M}(Z_0)\;.
$$
Notice that the direction of $\bb{M}(\psi)$ is reversed with respect to the direction of $\psi$ since $\bb{M}$ must be contravariant.
The map $\bb{M}(\phi_-)$ endows $\bb{M}(Z_0)$ with a $\Z_2$-action,
let  $\bb{M}(Z_0)^{\Z_2}\subseteq \bb{M}(Z_0)$ be the subset of invariant points under this action. Observing that $\psi\circ\phi_-=\psi$ one gets that $\bb{M}(\phi_-)\circ\bb{M}(\psi)=\bb{M}(\psi\circ\phi_-)=\bb{M}(\psi)$ one obtains that the image of $\bb{M}(Z_1)$ under $\bb{M}(\psi)$ is made by invariant points.
In summary, we showed that a  coefficient system $\bb{M}$ for $\Z_2$  consists of: (i) an abelian group $\bb{M}(Z_1)$; (ii) an abelian group $\bb{M}(Z_0)$ endowed with a $\Z_2$-action; (iii)
a homomorphism $\bb{M}(Z_1)\to \bb{M}(Z_0)^{\Z_2}$. 
In the following we will use the symbol $\bb{M}(Z_1)\rightsquigarrow \bb{M}(Z_0)$  as a synthetic description  of the coefficient system $\bb{M}$.
\hfill $\blacktriangleleft$
\end{example}

\subsection{Bredon equivariant cohomology}

Let $\n{G}$ be a finite group and $X$ a $\n{G}$-CW complex. Since the category $\mathtt{CoSy}_{\n{G}}$ of coefficient systems for $\n{G}$ is an abelian category, the $\mathrm{Hom}$-set
$$
\mathrm{Hom}_{\mathtt{CoSy}_{\n{G}}}(\bb{N},\bb{M})
$$
is an abelian group \cite[Theorem 2.39]{freyd-64} for any pair of coefficient systems $\bb{N}$ and $\bb{M}$. 
Now, let us take $\bb{N}$  be the coefficient system $\underline{C}_n(X)$ described in Example \ref{ex:main_coe:sys} and define
$$
C^n_{\n{G}}(X; \bb{M})\;: =\; \mathrm{Hom}_{\mathtt{CoSy}_{\n{G}}}(\underline{C}_n(X), \bb{M})\;.
$$
The natural transformation $\underline{\partial}_n : \underline{C}_n(X) \Rightarrow \underline{C}_{n-1}(X)$ induces a homomorphism
$$
\delta_n \;:\; C^n_{\n{G}}(X; \bb{M})\; \longrightarrow\; 
C^{n+1}_{\n{G}}(X; \bb{M})
$$
satisfying $\delta_{n+1}\circ\delta_n = 0$. This leads to a cochain complex $(C^\bullet_{\n{G}}(X; \bb{M}), \delta_\bullet)$.

\begin{definition}
Let $\n{G}$ be a finite group  and $X$ a $\n{G}$-CW complex. For any coefficient system $\bb{M} \in \mathtt{CoSy}_{\n{G}}$, the \textit{$n$-th Bredon $\n{G}$-equivariant cohomology of $X$ with coefficients in $\bb{M} $} is defined as the $n$-th cohomology of the cochain complex $(C^\bullet_{\n{G}}(X; \bb{M}), \delta_\bullet)$, \ie
$$
\s{H}^n_{\n{G}}(X, \bb{M})\;: =\; \mathrm{Ker}(\delta_n)/\mathrm{Im}(\delta_{n-1}).
$$
\end{definition}

\medskip

As a matter of fact, for any space $X$ with $\n{G}$-action, there is a $\n{G}$-CW complex which is $\n{G}$-equivariantly weakly homotopy equivalent to $X$ \cite{May}. Using such a $\n{G}$-CW complex, the Bredon equivariant cohomology of $X$ is defined. Given an invariant subspace $Y \subseteq X$, it is possible to introduce the relative cohomology $\s{H}^n_{\n{G}}(X| Y,\bb{M})$. Then the Bredon equivariant cohomology groups with coefficients in $\bb{M}$ constitute a $\n{G}$-equivariant generalized cohomology theory.

\begin{example}
In general, Bredon cohomology is difficult to calculate following  the definition. However, in a special case, one can compute it. For example, let $ \bb{M}$ be the coefficient system such that $\bb{M}(\n{G}/\n{H}) = 0$ for any proper subgroup $\n{H} \subset \n{G}$. Hence, only $\bb{M}(\n{G}/\n{G})$ can be non-trivial, and the homomorphisms $\bb{M}(\phi)$ are automatically determined. In particular, the $\n{G}$-action on $\bb{M}(\n{G}/\n{G})$ is trivial. In this case, the cochain complex $C^\bullet_{\n{G}}(X; \bb{M})$ is identified with the cellular cochain complex of the fixed point set $X^\n{G}$, \ie
$$
C^\bullet_{\n{G}}(X; \bb{M})\; \simeq \;C^n(X^\n{G}; \bb{M}(\n{G}/\n{G}))\;.
$$
It follows that $\s{H}^n_{\n{G}}(X,\bb{M}) \simeq H^n(X^\n{G}, \bb{M}(\n{G}/\n{G}))$. 
\hfill $\blacktriangleleft$
\end{example}

\subsection{Eilenberg-Mac Lane spaces}\label{sec:EMLS}
The ordinary cohomology $H^n(X, \n{A})$ of a CW-complex $X$ with coefficients in an abelian group $\n{A}$ can be represented as 
\[
H^n(X, \n{A}) \;\simeq \; [X, K(\n{A}, n)]\;,
\]
where $K(\n{A}, n)$ is the \emph{Eilenberg-Mac Lane space} of type $(\n{A}, n)$. There exists a parallel representation for the Bredon equivariant cohomology. For simplicity, let $n\geqslant 1$. Given a coefficient system $\bb{M} : \mathtt{Orb}_{\n{G}} \to \mathtt{Ab}$, the Eilenberg-Mac Lane space of type $(\bb{M}, n)$ is a path connected $\n{G}$-space $K(\bb{M}, n)$ such that 
\[
\underline{\pi}_k(K(\bb{M}, n))\;=\;
\left\{
\begin{array}{ll}
\bb{M} & (k=n) \\
0 & (k \neq n)\;,
\end{array}
\right.
\]
 where $\underline{\pi}_k$ has been described in Example \ref
{ex-homot},
and the  base point of the homotopy groups is chosen in the set of
 fixed points of $K(\bb{M}, n)$  (therefore $K(\bb{M}, n)^\n{G} \neq \emptyset$ is assumed implicitly). It is known \cite[Section II.6]{Bre} that there exists  a unique $\n{G}$-space,  up to $\n{G}$-equivariant (weak) homotopy equivalence, which provides the identification
\begin{equation}\label{eq:Eil-MC}
\s{H}^n_\n{G}(X, \bb{M})\; \simeq\; [X, K(\bb{M}, n)]_{\n{G}}\;,
\end{equation}
where the symbol in the right hand side denotes the set of $\n{G}$-equivariant homotopy classes of $\n{G}$-equivariant  maps $X \to K(\bb{M}, n)$.
{By construction, for every $\n{G}$-equivariant map $f : X \to K(\bb{M}, n)$ there exists an element 
$\chi^n(f)\in \s{H}^n_\n{G}(X, \bb{M})$ which is the image of the class $[f]\in [X, K(\bb{M}, n)]_{\n{G}}$ under the isomorphism \eqref{eq:Eil-MC}.
The element $\chi^n(f)$ is called the \emph{characteristic class} of $f$ according to the definition given in \cite[Section II.3]{Bre}.
If ${\rm id}:K(\bb{M}, n)\to K(\bb{M}, n)$ is the identity map, then the corresponding element $\iota:=\chi^n({\rm id})$ will be called the \emph{universal Bredon class}. In view of \cite[Chapter II, Proposition 3.2]{Bre} one has that
\[
\chi^n(f)\;=\;\chi^n(f\circ {\rm id})\;=\;f^*\circ \chi^n( {\rm id})\;=\;f^*\circ \iota\;.
\]
Therefore, the isomorphism \eqref{eq:Eil-MC} is realized by $[f]\mapsto f^*\circ \iota$.}

{
\begin{remark}\label{rk:gen_univB}
In the special case $\s{H}^n_\n{G}(K(\bb{M}, n); \bb{M})\simeq\Z$, the  defining property of the universal Bredon class implies that $\iota$ must be
a generator identifiable with $+1$ or $-1$. This can be proved by contradiction. Let $\alpha$ be an isomorphism between $\s{H}^n_\n{G}(K(\bb{M}, n); \bb{M})$ and $\Z$, and assume that $m:=\alpha(\iota)\neq \pm1$. For $\iota'\in \s{H}^n_\n{G}(K(\bb{M}, n); \bb{M})$ such that $\alpha(\iota')=1$, there exists a $\n{G}$-equivariant map $f : K(\bb{M}, n) \to K(\bb{M}, n)$ such that $f^*\circ \iota=\iota'$. Consequently, one can construct the homomorphism $\alpha\circ f^*\circ\alpha^{-1}:\Z\to\Z$
which maps $m$ to $1$. But this contradicts the fact that $m\neq\pm 1$.  
\hfill $\blacktriangleleft$
\end{remark}
}

\begin{example}
Let $\n{A}$ be an abelian group, and $\underline{\n{A}}$ the \emph{constant} coefficient system such that $\underline{\n{A}}(\n{G}/\n{H}) =\n{A}$ for all $\n{G}/\n{H} \in \mathtt{Orb}_{\n{G}} $, and $\underline{\n{A}}(\phi)$ is the identity map for any morphism $\phi$ in $\mathtt{Orb}_{\n{G}}$. Let $K(\n{A}, n)$ be the usual Eilenberg-Mac Lane space of type $(\n{A}, n)$, namely  a path connected space  such that $\pi_k(K(\n{A}, n))\simeq \n{A}$ if $k=n$, and $\pi_k(K(\n{A}, n))=0$ otherwise. Let $\n{G}$ act on $K(\n{A}, n)$ trivially. Then $K(\n{A}, n)^\n{H} = K(\n{A}, n)$ for any subgroup $\n{H} \subseteq\n{ G}$. Hence, the $\n{G}$-space $K(\n{A}, n)$ realizes $K(\underline{\n{A}}, n)$. It follows that
$$
\s{H}^n_\n{G}(X, \underline{\n{A}})\;
\simeq\; [X, K(\underline{\n{A}}, n)]_{\n{G}}\;
\simeq\; [X/\n{G}, K(\underline{\n{A}}, n)]\;\simeq\; H^n(X/\n{G}, \n{A})
$$
for any $\n{G}$-space $X$. In other words, the Bredon equivariant cohomology of a $\n{G}$-space $X$ with coefficients in $\underline{\n {A}}$ agrees with the ordinary cohomology of the quotient space $X/\n{G}$ with coefficients in $\n{A}$.
\hfill $\blacktriangleleft$
\end{example}

\subsection{Atiyah-Hirzebruch spectral sequence}

Let $X$ be a finite CW complex with base point $\ast$, and $h^\bullet = \{ h^n \}_{n \in \Z}$ a generalized cohomology theory satisfying the
Eilenberg-Steenrod  axioms \cite[Section I.3]{eilenberg-steenrod-52}. Then the Atiyah-Hirzebruch spectral sequence is a spectral sequence such that its $E_2$-term is
$$
E_2^{p, q} = H^p(X, h^q(\{\ast\}))
$$
and converges to a graded quotient of $h^\bullet(X)$ \cite[Chapter 9]{spanier-66} or  \cite[Chapter 9]{davis-kirk-01}. The construction of the Atiyah-Hirzebruch spectral sequence can be generalized 
to $\n{G}$-CW complexes, and one gets:
\begin{theorem}[{\cite[Section IV.4]{Bre}}]\label{theo_spec_seq}
Let $\n{G}$ be a finite group, $X$ a finite $\n{G}$-CW complex, and $h^\bullet = \{ h^n \}_{n \in \Z}$ a $\n{G}$-equivariant generalized cohomology theory. Then there is an Atiyah-Hirzebruch spectral sequence such that its $E_2$-term is
$$
E_2^{p, q}\; =\; \s{H}^p_{\n{G}}(X, \underline{h}^q)\;,
$$
where the coefficient system $\underline{h}^q$ has been described in Example \ref{ex:hh},  
and which converges to a graded quotient of $h^\bullet(X)$.
\end{theorem}

\begin{example}
As a trivial application, let us take a generalized cohomology theory $h^\bullet$, and define an equivariant generalized cohomology theory $h_{\n{G}}^\bullet$ by $h^n_{\n{G}}(X): = h^n(X/\n{G})$ through the quotient $X/\n{G}$ of the $G$-CW complex $X$. The coefficient system $\underline{h}^\bullet_G$ turns out to be the constant coefficient system $\underline{h^n(\{\ast\})}$. Then the $E_2$-term of the spectral sequence  is
$$
E_2^{p, q}\; =\; \s{H}^p_{\n{G}}(X, \underline{h}_{\n{G}}^q)
\;=\; \s{H}^p_{\n{G}}(X, \underline{h^n(\{\ast\})})
\;=\; {H}^p(X/\n{G}, h^q(\{\ast\})),
$$
and reduces to the usual Atiyah-Hirzebruch spectral sequence of the orbit $X/\n{G}$.
\hfill $\blacktriangleleft$
\end{example}

\section{The FKMM invariant and Bredon cohomology}
\label{sec:FKMM}

This section is devoted to the study of the FKMM invariant from the viewpoint of the Bredon equivariant cohomology. This will provide the framework for the proof of the main Theorem \ref{thm:main_in_the_body}.

\subsection{Borel and Bredon equivariant cohomology}
The first necessary step is to relate the
 Borel equivariant cohomology with the Bredon equivariant cohomology.
In order to help the reader, a short review about the Borel equivariant  cohomology is presented in Appendix \ref{subsec:borel_cohom}. In the following we will focus on the case of a $\Z_2$-space $X$ with involution $\tau$. The set of the fixed point will be denoted with $X^\tau$.
A $\Z_2$-CW pair $(X, Y)$ is given by a  $\Z_2$-CW complex $X$ a $\Z_2$-CW subcomplex $Y$ and a 
sub-complex inclusion $Y\hookrightarrow X$. In particular one has that $(X, X^\tau)$ is a $\Z_2$-CW pair. Given a $\Z_2$-CW pair $(X, Y)$, we will denote with $H^\bullet_{\Z_2}(X| Y,\s{Z})$
 the equivariant Borel  cohomology of $X$ relative to $Y$  with local coefficients in $\s{Z}$. 
 By convention we will fix $H^n_{\Z_2}(X| Y,\s{Z})=0$ for every $n<0$.

\medskip

\begin{lemma}\label{lemma:eq_coho-the}
For any $\Z_2$-CW pair $(X, Y)$ and $n \in \Z$, we use the Borel equivariant cohomology to define
$$
\begin{aligned}
h^n_{j}(X, Y)\; &:=\;  H^n_{\Z_2}(X| Y \cup X^\tau, \Z(j))\;,\qquad j=0,1\;.\\
\end{aligned}
$$
Then $h_{0}^\bullet: = \{ h^n_{0} \}_{n \in \Z}$ and $h_{1}^\bullet: = \{ h^n_{1} \}_{n \in \Z}$ are $\Z_2$-equivariant generalized cohomologies. 
\end{lemma}

\begin{proof}
The Borel equivariant cohomology is an equivariant generalized cohomology
\cite[Theorem 1.2.6]{allday-puppe-93}, meaning that it satisfies  the equivariant version  of  the Eilenberg-Steenrod axioms  of  a generalized cohomology theory.
 Thus, for any $\Z_2$-CW complex $X$ and its subcomplexes $Y, Y' \subset X$, one has the exact sequence for the triad $(X, Y, Y')$
$$
\cdots 
H^n_{\Z_2}(X| Y \cup Y',\s{Z})\; \to\;
H^n_{\Z_2}(X| Y',\s{Z}) \;\to\;
H^n_{\Z_2}(Y| Y \cap Y',\s{Z}) \;\to\;
H^{n+1}_{\Z_2}(X| Y \cup Y',\s{Z}) \;
\cdots
$$
where the isomorphism
$$
H^n_{\Z_2}(Y\cup  Y'|  Y',\s{Z}) \;\simeq\;H^n_{\Z_2}(Y| Y \cap Y',\s{Z}) 
$$
has been used in third place.
Setting  $Y' = X^\tau$ and $\s{Z}=\Z(j)$ with $j=0$ or $j=1$, one gets the exact sequence 
$$
\cdots \;\to\;
h^n_{j}(X, Y) \;\to\;
h^n_{j}(X,\emptyset) \;\to\;
h^n_{j}(Y,\emptyset) \;\to\;
h^{n+1}_{j}(X, Y) \;\to\;
\cdots
$$
which shows that $h^\bullet_{j}$ is subject to the exactness axiom. 
The other axioms for the equivariant generalized cohomology theory for $h^\bullet_{j}$ follow from those of $H_{\Z_2}^\bullet$. 
\end{proof}

\medskip

The group $\Z_2$ can act on $\Z$ in two ways. There is the \emph{trivial} action induced by  $- 1 \in \Z_2$ on $n \in \Z$ by $-1: n \mapsto n$, and there is the \emph{flip} action given by $-1: n \mapsto -n$. We will write $\widetilde{\Z}$ for the group $\Z$ endowed with the flip action and we will use the symbol  
$\Z$ for the case of the trivial action. Using the notation introduced in Example \ref{ex:coef_sysZ2} we will introduce three coefficient systems for $\Z_2$. The first one, denoted with $0\rightsquigarrow 0$, consists in identifying both  the orbits $Z_0$ and $Z_1$ with the trivial group $0$.
In the second coefficient system $0\rightsquigarrow \Z$ the orbit  
$Z_0$ is again identified with $0$ while the orbit $Z_1$ is identified with $\Z$ (trivial action). The last coefficient system $0\rightsquigarrow \widetilde{\Z}$ is obtains as the previous one, but now $Z_1$ is identified with $\widetilde{\Z}$ (flip action).

\begin{lemma}\label{lemma:ident_coesys}
Let $h_{j}^\bullet$, with $j=0,1$,  be the equivariant cohomology theory defined in Lemma \ref{lemma:eq_coho-the} and $\underline{h}_{j}^\bullet$ the associated coefficient systems constructed according to 
the procedure described in Example \ref{ex:hh}.
Then, one has that 
\begin{align*}
\underline{h}_{0}^n
&\;=\;
\left\{
\begin{array}{ll}
0 \rightsquigarrow \Z & (n = 0) \\
0 \rightsquigarrow 0& (n \neq 0)\;,
\end{array}
\right.
&
\underline{h}_{1}^n
&\;=\;
\left\{
\begin{array}{ll}
0 \rightsquigarrow \widetilde{\Z} & (n = 0) \\
0 \rightsquigarrow 0 & (n \neq 0)\;.
\end{array}
\right.
\end{align*}
\end{lemma}

\begin{proof}
First of all let us evaluate  $\underline{h}^n_{j}$ on the orbit $Z_1=\Z_2/\Z_2=\{z\}$ which is a singleton. From its very definition one gets that 
$$
\underline{h}^n_{j}(Z_1)\;=\;{h}^n_{j}(\{z\},\emptyset)\;=\;
H^n_{\Z_2}(\{z\}| \{z\}, \Z(j))\;=\;0
$$
where the last (tautological) equality works independently of $n\in\Z$ and $j=0,1$.
As the second step let us evaluate $\underline{h}^n_{j}$ on the orbit $Z_0=\Z_2/\{+1\}=\{z_+,z_-\}$ which is a two-point space (identifiable with $\Z_2$) on which $\Z_2$ acts freely. For $j=0$ one has that
\begin{equation}\label{eq:h0}
\underline{h}^n_{0}(Z_0)
\;=\; h^n_{0}(Z_0,\emptyset)
\;=\; H^n_{\Z_2}(Z_0,\Z(0)) \;\simeq\;H^n_{\Z_2}(Z_0,\Z)\;.
\end{equation}
To compute the Borel cohomology we need to construct the homotopy quotient of $Z_0$ according to 
\eqref{eq:homot_quot}. However, since $\Z_2$ acts freely on $\Z_0$ it follows that the homotopy quotient is  homotopy equivalent to the regular quotient $Z_0/\Z_2\simeq\{\ast\}$ which is a sigleton. Then, starting from the last isomorphism in \eqref{eq:h0} one gets
\begin{equation}\label{eq:jja}
\underline{h}^n_{0}(Z_0)
 \;\simeq\;H^n_{\Z_2}(Z_0,\Z)  \;\simeq\;H^n(\{\ast\},\Z)\;=\;
\left\{
\begin{array}{ll}
\Z & (n = 0) \\
0 & (n \neq 0)
\end{array}
\right.\;.
\end{equation}
Moreover, 
the $\Z_2$-action on $\{\ast\}$ induces the trivial action on $H^0(\{\ast\},\Z)=\Z$. 
For $j=1$ the  equivalent of \eqref{eq:h0} reads
\begin{equation}\label{eq:h00}
\underline{h}^n_{1}(Z_0)
\;=\; h^n_{1}(Z_0,\emptyset)
\;=\; H^n_{\Z_2}(Z_0,\Z(1))\;.
\end{equation}
By using the exact sequence in \cite[Proposition 2.3]{gomi-15}
one obtains 
$$
\cdots\; 
H^n_{\Z_2}(Z_0,\Z)\; \to\;
H^n(Z_0,\Z) \;\to\;
H^n_{\Z_2}(Z_0,\Z(1))\; \to\;
H^{n+1}_{\Z_2}(Z_0,\Z)\;
\cdots
$$
In view of \eqref{eq:jja} one gets
$$
H^n_{\Z_2}(Z_0,\Z(1))\;\simeq\; H^n(Z_0,\Z)\;=\;0\;, \qquad n\neq 0
$$
and
$$
0\;\to\;\Z\; \to\;
\Z\oplus\Z \;\to\;
H^0_{\Z_2}(Z_0,\Z(1))\; \to\;
0\; 
$$
where the isomorphism $H^0(Z_0,\Z)\simeq\Z\oplus\Z$ for the two-point space $Z_0$
has been used. This implies that 
$$
H^0_{\Z_2}(Z_0,\Z(1))\;\simeq\; (\Z \oplus \Z)/\Z \;\simeq\;\Z
$$  
as an abelian group.  To identify the $\Z_2$-action on $H^0_{\Z_2}(Z_0,\Z(1))$ let us observe that  $\Z_2$ acts on $H^0_{\Z_2}(Z_0,\Z) \simeq \Z$ trivially as discussed after \eqref{eq:jja}. On the other hand the induced $\Z_2$-action on $H^0(Z_0,\Z) \cong \Z \oplus \Z$ is given by $(m, n) \mapsto (n, m)$. As a consequence the induced $\Z_2$-action on $H^0_{\Z_2}(Z_0,\Z(1))$ is the non-trivial one. In summary, we showed that
\begin{equation}\label{eq:jjaa}
\underline{h}^n_{1}(Z_0)
 \;\simeq\;
\left\{
\begin{array}{ll}
\widetilde{\Z} & (n = 0) \\
0 & (n \neq 0)\;.
\end{array}
\right.
\end{equation}
and this concludes the proof.
\end{proof}

\medskip

We are now in position to compare the Borel and Bredon equivariant cohomology groups.

\begin{theorem}\label{teo_nat_iso}
For any finite $\Z_2$-CW complex $X$ and $n \in \Z$, there are natural isomorphisms of groups
\begin{align*}
H^n_{\Z_2}(X| X^\tau, \Z(0))\;&\simeq\; \s{H}^n_{\Z_2}(X, 0 \rightsquigarrow \Z)\\
H^n_{\Z_2}(X| X^\tau,\Z(1))\; &\simeq \;\s{H}^n_{\Z_2}(X, 0 \rightsquigarrow \widetilde{\Z})\;.
\end{align*}
\end{theorem}

\begin{proof}
The Atiyah-Hirzebruch spectral sequence described in Theorem \ref{theo_spec_seq} provides
$$
E_2^{p, q}\; =\; \s{H}^p_{\Z_2}(X, \underline{h}^q_{j})\; \Rightarrow\;
h^\bullet_{j}(X)\;,\qquad j=0,1\;.
$$
By the description of the coefficient system $\underline{h}^\bullet_{j}$ given in Lemma \ref{lemma:ident_coesys}, one deduces that the spectral sequence degenerates at $E_2$, and this yields
$$
H^n_{\Z_2}(X| X^\tau,\Z(j))\;
=\; h^n_{j}(X,\emptyset)
\;\simeq\; E_\infty^{n, 0} \;\simeq\;E_2^{n, 0}
\;=\; \s{H}^n_{\Z_2}(X,\underline{h}^0_{j} )\;.
$$
The identification of $\underline{h}^0_{j}$ given in Lemma \ref{lemma:ident_coesys} concludes the proof.
\end{proof}

\subsection{The universal FKMM invariant}
Let us review here some cohomological aspects  of the universal FKMM invariant as described in \cite[Section 6]{denittis-gomi-14-gen}  and \cite[Section 2.6]{denittis-gomi-18-I}. For any positive integer $k\in\N$, let
$$
\s{B}_{2k}\;: =\; {\rm Gr}_{2k}(\C^\infty)\; =\; 
\varinjlim_n
\n{U}(2n)/(\n{U}(2k) \times \n{U}(2n-2k))
$$
be the Grassmannian (or the classifying space) of the unitary group $\n{U}(2k)$ in dimension $2k$. 
It is known that $\s{B}_{2k}$ is path-connected, which implies $\pi_0(\s{B}_{2k})=0$. Moreover, the homotopy groups of  $\s{B}_{2k}$ are related to the homotopy groups of $\n{U}(2k)$ by the formula
$\pi_n(\s{B}_{2k})=\pi_{n-1}(\n{U}(2k))$ \cite[eq. A.3]{denittis-gomi-18-III}. In particular this provides $\pi_1(\s{B}_{2k})=0$ and $\pi_2(\s{B}_{2k})\simeq\Z$, showing that $\s{B}_{2k}$ is also
simply connected. The cohomology ring of $\s{B}_{2k}$ is
\begin{equation}\label{eq:cohom_rin}
H^\bullet(\s{B}_{2k},\Z)\;\simeq\;\Z[\rr{c}_1,\ldots,\rr{c}_{2k}]\;,\qquad \rr{c}_j\;\in\: H^{2j}(\s{B}_{2k},\Z)\;,
\end{equation}
which is the ring of polynomials with integer coefficients and $2k$ generators $\rr{c}_j$ of even degree called universal Chern classes
 \cite[Appendix A]{denittis-gomi-18-III}.
Consider the involution $\rho : \n{U}(2k)\to \n{U}(2k)$ given by $\rho(U): = Q \overline{U} Q^{-1}$, where $Q \in \n{U}(2k)$ is the matrix
$$
Q \;:=\;
\left(
\begin{array}{cc}
0 & -{\bf 1}_{\C^k} \\
{\bf 1}_{\C^k} & 0
\end{array}
\right)
$$
and $\overline{U}$ denotes the complex conjugate of  $U$.
We also write $\rho : \s{B}_{2k} \to \s{B}_{2k}$ for the involution induced from the involution  on $\n{U}(2k)$. Let $\mathscr{T}^\infty_{2k} \to \s{B}_{2k}$ denote the universal (or tautological) vector bundle of rank $2k$ with total space $\mathscr{T}^\infty_{2k} :=\varinjlim_n \mathscr{T}^n_{2k}$
defined through
$$
\mathscr{T}^n_{2k}\;:=\;\{(\Sigma,{\rm v})\in {\rm Gr}_{2k}(\C^n)\times\C^n\;|\; {\rm v}\in\Sigma\}\;.
$$
Since $\mathscr{T}^\infty_{2k}$ serves as the universal ``Quaternionic'' vector bundle of rank $2k$, it turns out that the $\Z_2$-space $\s{B}_{2k}$ is the classifying space of ``Quaternionic'' vector bundles of rank $2k$ \cite[Theorem 2.4]{denittis-gomi-14-gen}. 
The \emph{universal FKMM invariant} is the FKMM invariant of the universal bundle, \ie
$$
\kappa_{\rm univ}\; :=\; \kappa(\mathscr{T}^\infty_{2k})\;
\in\; H^2_{\Z_2}(\s{B}_{2k}| \s{B}_{2k}^\rho,\Z(1))\;.
$$

\medskip

For any $\n{Z}_2$-CW complex $X$ and $n \in \Z$, the inclusion $j:(X, \emptyset) \hookrightarrow (X, X^\tau)$ induces a natural homomorphism 
$$
j^* \;:\; H^n_{\Z_2}(X| X^\tau,\Z(1))\; \longrightarrow\; H^n_{\Z_2}(X	,\Z(1))
$$
which fits into the exact sequence \eqref{eq:Long_seq} of the pair $(X, X^\tau)$. There is also a natural homomorphism
$$
f \;: \;H^n_{\Z_2}(X,\Z(1)) \;\longrightarrow\;  H^n(X,\Z)
$$
forgetting the $\Z_2$-action, which also fits into an exact sequence \cite[Proposition 2.3]{gomi-15}. If $X$ is path connected, then the Hurewicz homomorphism (in homology) induces a natural homomorphism
$$
H^n(X,\Z)\; \longrightarrow\; \mathrm{Hom}(\pi_n(X), \Z)\;.
$$

\begin{proposition}\label{prp:cohom1}
Let $k\in\N$ be a positive integer. Then,
there are isomorphisms
\[
\begin{aligned}
H^2_{\Z_2}(\s{B}_{2k}| \s{B}_{2k}^\rho,\Z(1)) \;&
\overset{j^*}{\simeq} \;H^2_{\Z_2}(\s{B}_{2k},\Z(1))
\;\overset{f}{\simeq} \;H^2(\s{B}_{2k},\Z)\\
&\simeq\; \mathrm{Hom}(\pi_2(\s{B}_{2k}), \Z)\;\simeq\;\Z\;.
\end{aligned}
\]
The universal FKMM invariant is related to the (universal) first Chern class $\rr{c}_1$ by the formula
$$
f(j^*(\kappa_{\rm univ}))\;=\;\rr{c}_1
$$ 
and provides a basis of $H^2_{\Z_2}(\s{B}_{2k}| \s{B}_{2k}^\rho,\Z(1))$.
\end{proposition}
\begin{proof}
The isomorphism $f$ has been proved in \cite[Proposition A.3]{denittis-gomi-18-I}, while the isomorphism  $j^*$ has been established in 
\cite[Proposition A.5]{denittis-gomi-18-I}. The isomorphism 
$H^2(\s{B}_{2k},\Z)\simeq \Z$ follows from \eqref{eq:cohom_rin}. In particular $H^2(\s{B}_{2k},\Z)$ is generated by the (universal) first Chern class $\rr{c}_1$.
The fact that $\kappa_{\rm univ}$ is proven in 
\cite[Proposition 2.15 (3)]{denittis-gomi-18-I}, along with its relation with the first $\rr{c}_1$. Since $\pi_1(\s{B}_{2k})\simeq\Z$ one has that $H_1(\s{B}_{2k})\simeq\Z$ in view of the Hurewicz homomorphism,
and in turn $H^2(\s{B}_{2k},\Z)\simeq\mathrm{Hom}(\pi_2(\s{B}_{2k}), \Z)$ as a consequence of the universal coefficient theorem.
\end{proof}

\subsection{A relevant Eilenberg-Mac Lane space}\label{sec:cohom_eilMC}
This subsection is aimed to show that the Eilenberg-Mac Lane space 
\[
\s{K}\;:=\;K(0 \rightsquigarrow \tilde{\Z}, 2)
\]
has the same cohomological nature as the classifying space $\s{B}_{2k}$. Let us recall from Section \ref{sec:EMLS} that $\s{K}$ is a path connected space endowed with a $\Z_2$-action denoted with $\tau$, such that $\s{K}^\tau\neq\emptyset$ and
\begin{equation}\label{eq:Eil-MC-2}
\s{H}^2_{\n{Z}_2}(X, 0 \rightsquigarrow \tilde{\Z})\; \simeq\; [X, \s{K}]_{\n{Z}_2}\;,
\end{equation}
for any $\Z_2$-CW complex $X$.

\medskip

The relevant cohomology properties of $\s{K}$ are summarized in the following result which parallels the property of the space  $\s{B}_{2k}$ proved in Proposition \ref{prp:cohom1}.

\begin{proposition}\label{prop:k_isod}
There are isomorphisms
\begin{align*}
H^2_{\Z_2}(\s{K}| \s{K}^\tau,\Z(1))\; 
\overset{j^*}{\simeq}\; H^2_{\Z_2}(\s{K},\Z(1))
\;\overset{f}{\simeq}\; H^2(\s{K},\Z)
\;\simeq\; \mathrm{Hom}(\pi_2(\s{K}), \Z)\;\simeq\; \Z\;.
\end{align*}
\end{proposition}

\medskip

This result is proved in several technical steps mimicking the strategy used
 in \cite[Appendix A]{denittis-gomi-18-I}. More precisely the proof is contained in Lemmas \ref{lemm:homot_typ}, \ref{lemma_tool1}, \ref{lemma_tool2} and \ref{lemma_tool3}.

\subsection{Bredon cohomology interpretation of the The FKMM invariant}
We are now in position to
carry out the  study of the FKMM invariant from the viewpoint of the Bredon cohomology. Recall that Theorem \ref{teo_nat_iso} establishes
the natural isomorphism
\[
H^2_{\Z_2}(X| X^\tau,\Z(1))\; \simeq\; \s{H}^2_{\Z_2}(X,0 \rightsquigarrow \widetilde{\Z})
\]
for any $\Z_2$-CW complex. We will freely use this isomorphism, and for that we can identify the universal FKMM invariant $\kappa_{\rm univ}$ with an element in the Bredon cohomology, \ie 
\[
\kappa_{\rm univ}\;\in\; \s{H}^2_{\Z_2}(\s{B}_{2k},0 \rightsquigarrow \widetilde{\Z})\;.
\]
{Let us recall the \emph{universal Bredon class} 
\[
\iota\; \in\; \s{H}^2_{\Z_2}(\s{K},0 \rightsquigarrow \widetilde{\Z})\;\simeq\;
H^2_{\Z_2}(\s{K}| \s{K}^\tau,\Z(1))\;\simeq\;\Z
\]
as defined in Section \ref{sec:EMLS}. In view of Remark \ref{rk:gen_univB} it follows that  $\iota$ is a generator.
By its nature, there exists a $\Z_2$-equivariant map
\begin{equation}\label{eq:eq_map}
\varphi\; :\; \s{B}_{2k} \;\longrightarrow\; \s{K}
\end{equation}
such that $\varphi^*(\iota) = \kappa_{\rm univ}$.}

\begin{lemma} \label{lem:4_equivalence}
Let $\varphi$ be the $\Z_2$-equivariant map \eqref{eq:eq_map}. Then,
for any integer $k\in\N$, it holds true that:
\begin{itemize}
\item[(i)]
The homomorphism $\varphi_* : \pi_n(\s{B}_{2k}) \to \pi_n(\s{K})$ is bijective for $n \le 3$ and  surjective for $n = 4$;
\vspace{1mm}
\item[(ii)]
The homomorphism $\varphi_* : \pi_n(\s{B}_{2k}^\rho) \to \pi_n(\s{K}^\tau)$ is bijective for $n \le 3$ and surjective for $n = 4$. 

\end{itemize}
\end{lemma}

\begin{proof}
(i) The homotopy of the Eilenberg-Mac Lane space $\s{K}$ is given by \eqref{eq_homot}. Moreover, it is known 
 that  $\pi_n(\s{B}_{2k}) \simeq \pi_{n-1}(\n{U}(2k))$ \cite[eq. A.3]{denittis-gomi-18-III}.  In summary, one has that
\[
\begin{array}{|c||c|c|c|c|c|}
\hline
& n = 0 & n = 1 & n = 2 & n = 3 & n = 4 \\
\hline
\hline
\pi_n(\s{B}_{2k}) & 0 & 0 & \Z & 0 & \Z \\
\hline
\pi_n(\s{K}) & 0 & 0 & \Z & 0 & 0 \\
\hline
\end{array}
\]
Therefore, the claim  (i) is evident for $n \neq 2$. For $n = 2$, in view of Propositions \ref{prp:cohom1} and \ref{prop:k_isod},
one has the commutative diagram
\[
\begin{array}{c@{}c@{}c@{}c@{}c@{}c@{}c}
H^2_{\Z_2}(\s{B}_{2k}| \s{B}_{2k}^\rho,\Z(1)) & \; \simeq \; &\;
\mathrm{Hom}(\pi_2(\s{B}_{2k}), \Z)\;\simeq\;\Z \\
\uparrow & & \uparrow  \\
H^2_{\Z_2}(\s{K}| \s{K}^\tau,\Z(1)) &\; \simeq\; &
\mathrm{Hom}(\pi_2(\s{K}), \Z)\;\simeq\;\Z
\end{array}
\]
where all  vertical maps are induced from the equivariant map $\varphi$. Because $\varphi^*$ relates the generators $\iota$ and $\kappa_{\rm univ}$, it follows that the left most homomorphism
\[
\varphi^* \;:\; \underbrace{H^2_{\Z_2}(\s{K}| \s{K}^\tau,\Z(1))}_{\Z}\; \longrightarrow\; 
\underbrace{H^2_{\Z_2}(\s{B}_{2k}| \s{B}_{2k}^\rho,\Z(1))}_{\Z}
\]
is indeed an isomorphism. Therefore also the right most homomorphism, which is induced from $\varphi_* : \pi_2(\s{B}_{2k}) \to \pi_2(\s{K})$, must be an isomorphism. This implies that $\pi_2(\s{B}_{2k})\simeq \pi_2(\s{K})$.\\
(ii) By design, $\s{K}^\tau \neq \emptyset$, and equation \eqref{eq_homot00} provides $\pi_n(\s{K}^\tau) = 0$ for all $n$. On the other hand 
\[
\s{B}_{2k}^\rho\;\simeq\;  {\rm Gr}_{k}(\n{H}^\infty)\; =\; 
\varinjlim_n
{\rm Sp}(n)/({\rm Sp}(k) \times {\rm Sp}(n-k))
\]
 where $\n{H}$ denotes the (non-commutative) field of quaternions and ${\rm Sp}(k)$ is the compact symplectic group. Therefore, from the the homotopy exact sequence one obtains that $\pi_n(\s{B}_{2k}^\rho)\simeq\pi_{n-1}({\rm Sp}(k))$.  The homotopy groups of $\s{B}_{2k}^\rho$ and $\s{K}^\tau$ in low degrees are shown in the following table:
\[
\begin{array}{|c||c|c|c|c|c|}
\hline
& n = 0 & n = 1 & n = 2 & n = 3 & n = 4 \\
\hline
\pi_n(\s{B}_{2k}^\rho) & 0 & 0 & 0 & 0 & \Z \\
\hline
\pi_n(\s{K}^\tau) & 0 & 0 & 0 & 0 & 0 \\
\hline
\end{array}
\]
The proof follows directly from the values listed above. 
\end{proof}

\medskip

We are now in position to prove our main result.

\begin{theorem} \label{thm:main_in_the_body}
Let $k\in\N$ be a positive integer. Then, the FKMM invariant
$$
\kappa \;:\; \mathrm{Vect}^{2k}_{\rr{Q}}(X)\;\longrightarrow\; H^2_{\Z_2}(X| X^\tau,\Z(1))
$$
is bijective for any $\Z_2$-CW complex $X$ of dimension $d\leqslant 3$.
\end{theorem}\begin{proof}
For any $\Z_2$-CW complex $X$, the equivariant map \eqref{eq:eq_map} induces a natural map
\[
\varphi_* \;:\; [X, \s{B}_{2k}]_{\Z_2}\; \longrightarrow\; [X,\s{K}]_{\Z_2}\;.
\]
By the Whitehead theorem in equivariant homotopy theory \cite[Theorem 3.2]{May} and Lemma \ref{lem:4_equivalence}, the above map is bijective if the dimension of $X$ is less than or equal to three. Since $\s{B}_{2k}$ classifies ``Quaternionic'' vector bundles of rank $2k$, we have
\[
[X, \s{B}_{2k}]_{\Z_2}\; \cong\; \mathrm{Vect}^{2k}_{\rr{Q}}(X)\;,
\]
while $\s{K}$ represents the Bredon cohomology
\[
[X, \s{K}]_{\Z_2}\; \simeq\; 
\s{H}^n_{\Z_2}(X, 0 \rightsquigarrow \widetilde{\Z}) \;\simeq\; H^2_{\Z_2}(X|X^\tau,\Z(1))\;.
\]
It remains to verify that the composition of these bijections
\[
\mathrm{Vect}^{2k}_{\rr{Q}}(X)\; \simeq\;
[X, \s{B}_{2k}]_{\Z_2}\; \overset{\varphi_*}{\simeq}\;
[X, \s{K}]_{\Z_2}\; \simeq\;
H^2_{\Z_2}(X| X^\tau,\Z(1))
\]
is the FKMM invariant. Suppose that $\bb{E} \to X$ is a ``Quaternionic'' vector bundle of rank $2k$. Then there is an equivariant map $f : X \to \s{B}_{2k}$ such that $f^*\mathscr{T}^\infty_{2k} \simeq \bb{E}$. The assignment $[\bb{E}] \mapsto [f]$ realizes the first bijection. Composing $\varphi$, we get an equivariant map $\varphi \circ f : X \to \s{K}$. The assignment $[f] \mapsto [\varphi \circ f]$ realizes the second bijection. Finally, by the third bijection, we get a cohomology class $(\varphi \circ f)^*(\iota) \in H^2_{\Z_2}(X| X^\tau,\Z(1))$. By design, one has
\[
(\varphi \circ f)^*(\iota)\;
=\; f^*(\varphi^*(\iota))
\;=\; f^*(\kappa_{\rm univ})\;=\;f^*(\kappa(\mathscr{T}^\infty_{2k}))
\;=\; \kappa(f^*\mathscr{T}^\infty_{2k})
\;=\; \kappa(\bb{E})\;,
\]
where the second-last equality is justified by the naturality of $\kappa$.
\end{proof}

\medskip

\begin{remark}
By the Whitehead theorem in equivariant homotopy theory, the map $\varphi_* : [X, \s{B}_{2k}]_{\Z_2} \to [X, \s{K}]_{\Z_2}$ is surjective for any $\Z_2$-CW complex of dimension $d\leqslant 4$. It follows that the FKMM invariant $\kappa : \mathrm{Vect}^{2k}_Q(X) \to H^2_{\Z_2}(X| X^\tau,\Z(1))$ is also surjective in dimension $d=4$.
\hfill $\blacktriangleleft$
\end{remark}

\section{Equivariant cohomology of lens space}
\label{sec:lens}
The main aim of  this part is to clarify the wrong part in \cite[Section 5]{denittis-gomi-18-II} and provide corrected claims. 

\subsection{Setup and relevant results}\label{sec_setupG}
The three dimensional sphere can be parametrized as the unit sphere in $\C^2$,
\begin{equation}\label{eq:sphere_complex}
\n{S}^3\;\equiv\;\big\{(z_0,z_1)\in\C^2\ |\ |z_0|^2+|z_1|^2=1  \big\}\;\subset\;\C^2\;.
\end{equation}
This representation allows 
 $u\in\n{U}(1)$ to act (on the left) on $\n{S}^3$ through the mapping $(z_0,z_1)\mapsto (uz_0,uz_1)$. 
This action of $\n{U}(1)$ on  $\n{S}^3$ is evidently \emph{free}.
 The inclusion of $\Z_p\subset \n{U}(1)$, given by the fact that $\Z_p$ can be identified with the set of the $p$-th roots of the unity, implies that one can  define a free action of every cyclic group $\Z_p$ on $\n{S}^3$. More precisely we can let
 $k\in\Z_p$ act on $\n{S}^3$ through the rotation 
 \[
 k\;:\;(z_0,z_1)\longmapsto \left(\expo{\ii2\pi\frac{k}{p}}z_0,\expo{\ii2\pi\frac{k}{p}}z_1\right)\;.
 \] 
 The quotient space
 \[
 L_p\;:=\;\n{S}^3/\Z_p
 \]
is called the (three-dimensional) \emph{lens space} (see \cite[Example 18.5]{bott-tu-82} or \cite[Example 2.43]{hatcher-02} for more details) and sometime is denoted with the symbol $L(1;p)$. 
By combining the facts that $\n{S}^3$ is simply connected and the $\Z_p$-action on $\n{S}^3$
is free one concludes that $\n{S}^3$ is the {universal cover} of $L_p$.

\medskip

The parametrization \eqref{eq:sphere_complex}
allows to equip 
$\n{S}^3\subset\C^2$
with the involution 
 induced by the complex conjugation 
$(z_0,z_1)\mapsto(\overline{z_0},\overline{z_1})$.
The computation
 \[
 \expo{\ii2\pi\frac{p-k}{p}}\;=\;\expo{-\ii2\pi\frac{k}{p}}\;=\;\overline{\expo{\ii2\pi\frac{k}{p}}}\;,\qquad\quad k\in\Z_p\;
 \]
shows that 
 $\n{Z}_p\subset\n{U}(1)$ 
 is preserved by the complex conjugation.
Therefore, the involution on $\n{S}^3$ descends to an involution $\tau$ on
 $L_p$.  
  The involutive space $(L_p,\tau)$ inherits the structure of a smooth (three-dimensional) manifold with a smooth involution, hence it admits a $\Z_2$-CW-complex structure \cite[Theorem 3.6]{May}. Let us point out that it is possible
to think of  $L_p\to \C P^1$ as a \virg{Real}
principal $\n{U}(1)$-bundle where the \virg{Real} structure on the total space is provided by $\tau$, and the
 involution $\tau'$ on the base space $\C P^1$ is still given by the complex conjugation $\tau':[z_0,z_1]\mapsto[\overline{z_0},\overline{z_1}]$.
 
 \medskip

 Henceforth,
 let us focus now on the even case $p=2q > 0$.
 
 \medskip
 
 \noindent
 {\bf Fixed point set.} As proved in  \cite[Lemma 5.1]{denittis-gomi-18-II}
the fixed point set of $L_{2q}$ has the form
\[
L_{2q}^\tau\;=\;S_0\;\sqcup\;S_1\;\simeq\;\n{S}^1\;\sqcup\;\n{S}^1
\]
where
\begin{equation}\label{eq:conne_comp}
\begin{aligned}
S_0\;&:=\;\left.\left\{\big[\cos\theta,\sin\theta\big]\in L_{2q}\ \right|\ \theta\in\R\right\}\;,\\
S_1\;&:=\;\left.\left\{\left[\expo{-\ii\frac{\pi}{2q}}\cos\theta,\expo{-\ii\frac{\pi}{2q}}\sin\theta\right]\in L_{2q}\ \right|\ \theta\in\R\right\}\;.\\
\end{aligned}
\end{equation}
and $S_0\simeq S_1\simeq \n{S}^1$. In particular it is the disjoint union of two connected components.

 \medskip
 
 \noindent
 {\bf Equivariant Borel cohomology.} As proved in \cite[Section 5.2]{denittis-gomi-18-II} (see in particular Tables 2 and 4), one has that
\begin{equation}\label{eq:cohom:B_lens1}
{H}^1_{\Z_2}(L_{2q},\Z(1))\;\simeq\;\Z_2
\end{equation}
and
\begin{equation}\label{eq:cohom:B_lens2}
{H}^1_{\Z_2}(L_{2q}^\tau,\Z(1))\;=\;{H}^1_{\Z_2}(S_0,\Z(1))\;\oplus\;{H}^1_{\Z_2}(S_1,\Z(1))\;\simeq\; \Z_2\;\oplus\;\Z_2\;.
\end{equation}
In view of the first isomorphism in \eqref{eq:iso:eq_cohom}, the generators of these two groups can be represented by equivariant maps. More precisely the generator in \eqref{eq:cohom:B_lens1} is represented by the constant map on $L_{2q}$ with value $-1$.
Similarly, the two generators in \eqref{eq:cohom:B_lens2} are represented by the two constant maps on $S_0\sqcup S_1$ 
which take the value $+1$ on one connected component and $ -1 $ on the other.

 \medskip
 
 \noindent
 {\bf ``Real'' line bundles.} From \cite[Remark 5.2]{denittis-gomi-18-II}
 we know that  the Picard group of $L_{2q}$ and the 
\virg{Real} Picard group of the involutive space $(L_{2q},\tau)$ coincide and are given by
\[
{\rm Pic}_{\rr{R}}(L_{2q},\tau)\;\stackrel{c^{\rr{R}}_1}{\simeq}\;H^2_{\Z_2}(L_{2q},\Z(1))
\;\simeq\;\Z_{2q}\;\simeq\;
H^2_{\Z_2}(X,\Z) \;\stackrel{c_1}{\simeq}\;{\rm Pic}(L_{2q})\;.
\]
Therefore, there are only $2q$ complex line bundles over  $L_{2q}$ (up to isomorphisms), and each one of these  can be endowed with a unique (up to isomorphisms)  ``Real''  structure (\cf eq. \eqref{eq:iso:eq_cohom}). The representatives of these line bundles can be constructed explicitly. For $k\in\Z$, we let $u\in\Z_{2q}$ act on $\n{S}^3\times\C$ by $((z_0,z_1),\lambda)\mapsto ((uz_0,uz_1),\overline{u}^k\lambda)$.
Since the action is free on the base space the quotient defines a complex line bundle  
$\bb{L}_k\to L_{2q}$ (\cf \cite[Proposition 1.6.1]{atiyah-67}). From the construction it results evident that $\bb{L}_k=\bb{L}_{k+2q}$ and $\bb{L}_0=L_{2q}\times \C$ is the trivial line bundle. Moreover, $\bb{L}_1$ provides a basis for ${\rm Pic}(L_{2q})\simeq\Z_{2q}$ in view of the fact that ${\bb{L}_1}^{\otimes k}\simeq \bb{L}_k$. The  ``Real'' structure on $\bb{L}_k$ is evidently induced by the complex conjugation $\Theta:[(z_0,z_1),\lambda]\mapsto[\tau{(z_0,z_1)},\overline{\lambda}]=[\overline{(z_0,z_1)},\overline{\lambda}]$.
 The isomorphism class of $\bb{L}_1$ as a ``Real'' line bundle provides a generator for $H^2_{\Z_2}(L_{2q},\Z(1))$.

\subsection{Restricted ``Real" line bundles}
In \cite[Section]{denittis-gomi-18-II}, the restriction of the ``Real'' line bundle $\bb{L}_1\to L_{2q}$ to the fixed point set $L_{2q}^\tau = S_0 \sqcup S_1$ is studied. In particular in  \cite[Lemma 5.3]{denittis-gomi-18-II} it is claimed that the restrictions $\bb{L}_1|_{S_0}$ and $\bb{L}_1|_{S_1}$ are trivial. However, this result is \emph{wrong}. The argument used in the proof of \cite[Lemma 5.3]{denittis-gomi-18-II} is based on the construction of two nowhere vanishing invariant sections $s_j : S_j \to \mathscr{L}_1|_{S_j}$, $j=1,2$, but the constructed sections are actually ill-defined. The correct claim is contained in the following result.
\begin{lemma} \label{lem:restriction}
Both ``Real'' line bundles $\mathscr{L}|_{S_0}$ and $\mathscr{L}|_{S_1}$
 are non-trivial.
\end{lemma}
\begin{proof}
Let us start with an equivalent, but more  ``appropriate'', description of the connected components $S_0$ and $S_1$ of the fixed point set $L_{2q}^\tau$ given by
\begin{equation}\label{eq:conne_comp2}
\begin{aligned}
S_0\;&:=\;\left.\left\{\left[\cos\frac{\theta}{2},\sin\frac{\theta}{2}\right]\in L_{2q}\ \right|\ \theta\in\R/2\pi\Z\right\}\;,\\
S_1\;&:=\;\left.\left\{\left[\expo{-\ii\frac{\pi}{2q}}\cos\frac{\theta}{2},\expo{-\ii\frac{\pi}{2q}}\sin\frac{\theta}{2}\right]\in L_{2q}\ \right|\ \theta\in\R/2\pi\Z\right\}\;.\\
\end{aligned}
\end{equation}
The comparison with \eqref{eq:conne_comp} is obtained by observing that the action of $q\in\Z_{2q}$ on $\n{S}^3$ is obtained by multiplying 
the component with $\expo{\ii2\pi\frac{q}{2q}}=-1$, \ie by $(z_0,z_1)\mapsto(-z_0,-z_1)$. Therefore, it is sufficient to parametrize the components  $S_0$ and $S_1$ only with angles contained in the range $[0,\pi)$, or equivalently by using the parameter $\theta/2$.
In view of \eqref{eq:conne_comp2}, one can identify 
$S_0$ and $S_1$ with $\n{S}^1 \simeq \R/2\pi \Z$. Two nowhere-vanishing sections $s_j : S_j \to \mathscr{L}_1|_{S_j}$, for $j = 0, 1$, can be defined  by
\begin{align*}
s_0\left(\left[\cos \frac{\theta}{2}, \sin \frac{\theta}{2}\right]\right) 
&\;=\; \left[\left(\cos \frac{\theta}{2}, \sin \frac{\theta}{2}\right), \expo{-\ii \frac{\theta}{2}}\right]\;, \\
s_1\left(\left[\expo{-\ii\frac{\pi}{2q}}\cos \frac{\theta}{2}, 
\expo{-\ii\frac{\pi}{2q}}\sin \frac{\theta}{2}\right]\right) 
&\;=\;
\left[\left(\expo{-\ii\frac{\pi}{2q}}\cos \frac{\theta}{2}, 
\expo{-\ii\frac{\pi}{2q}}\sin \frac{\theta}{2}\right),\expo{\ii\frac{\pi}{2q}}\expo{-\ii \frac{\theta}{2}}\right] \;,
\end{align*}
for every $\theta\in \R/2\pi \Z$. A direct computation leads to
\begin{align*}
\Theta\left(s_0\left(\left[\cos \frac{\theta}{2}, \sin \frac{\theta}{2}\right]\right) \right) 
&=
s_0\left(\tau\left(\left[\cos \frac{\theta}{2}, \sin \frac{\theta}{2}\right]\right)\right)\cdot
\expo{\ii\theta}\;, \\
\end{align*}
where $\Theta$ is the ``Real'' structure on $\bb{L}_1$ and
the notation $\cdot
\expo{\ii\theta}$ on the right denotes the multiplication by the phase $\expo{\ii\theta}$ only on the fiber. A similar result holds also for the section $s_1$ as it can be checked by a straightforward computation.
As a consequence the two nowhere-vanishing sections $s_0$ and $s_1$ are not ``Real'' since $\Theta\circ s_j\neq s_j\circ \tau$. However, each  section $s_j$ induces an isomorphism between $\mathscr{L}_1|_{S_j}\to S_j$ and the product line bundle $\n{S}^1 \times \C\to \n{S}^1$ (on the circle $\n{S}^1 \simeq \R/2\pi\Z$)  with  ``Real'' structure $(\theta, \lambda) \mapsto (\theta, \expo{\ii\theta}\overline{\lambda})$. The latter ``Real'' line bundle is non-trivial, and this 
concludes the proof. 
\end{proof}

\subsection{``Quaternionic" structures}\label{sec:Q_str_len}
The different ``Quaternionic" structures on $(L_{2q},\tau)$ have been constructed in \cite[Proposition 5.7]{denittis-gomi-18-II}. As a result one has that classification of  equivalence classes of rank $2m$ ``Quaternionic" vector bundles over $(L_{2q},\tau)$ is given by
$$
{\rm Vec}_{\rr{Q}}^{2m}\big(L_{2q},\tau\big)\;\simeq\;\Z_{2q}\;,\qquad\quad \forall\ m\in\N\;. 
$$
In view of Theorem \ref{thm:main_in_the_body}, this fact must imply
$H^2_{\Z_2}(L_{2q}| L_{2q}^\tau,\Z(1))\simeq \Z_{2q}$. In fact this is the correct result, while the claim in  \cite[Proposition 5.4]{denittis-gomi-18-II} turns out to be \emph{wrong} as a consequence of the incorrectness of \cite[Lemma 5.3]{denittis-gomi-18-II}. 

\medskip

In order  to compute directly $H^2_{\Z_2}(L_{2q}| L_{2q}^\tau,\Z(1))$, let us make use of the exact sequence \eqref{eq:Long_seq}.  One gets that
\[
\begin{aligned}
\underbrace{H^1_{\Z_2}(L_{2q},\Z(1))}_{\Z_2} \overset{i^*}{\to}
\underbrace{H^1_{\Z_2}(L_{2q}^\tau,\Z(1))}_{\Z_2 \oplus \Z_2}& \to
H^2_{\Z_2}(L_{2q}| L_{2q}^\tau,\Z(1)) \to\\
& \to
\underbrace{H^2_{\Z_2}(L_{2q},\Z(1))}_{\Z_{2q}} \overset{i^*}{\to}
\underbrace{H^2_{\Z_2}(L_{2q}^\tau,\Z(1))}_{\Z_2 \oplus \Z_2},
\end{aligned}
\]
where the values of the cohomology groups are taken from Tables 
2 and 4 in \cite[Section 5.2]{denittis-gomi-18-II}. The homomorphism $i^* : H^1_{\Z_2}(L_{2q},\Z(1)) \to H^1_{\Z_2}(L_{2q}^\tau,\Z(1))$ induced from the inclusion $i : L_{2q}^\tau \hookrightarrow L_{2q}$ coincides with the diagonal map in view of the explicit description of the generators given in Section \ref{sec_setupG}. Therefore $i^* $ is injective, and its
 cokernel is $\Z_2$. 
 Let us identify $\bb{L}_1$ with the generator of $H^2_{\Z_2}(L_{2q},\Z(1))$ and observe that $i^*(\bb{L}_1)=\bb{L}_1|_{S_0\sqcup S_1}$ is not trivial
by Lemma \ref{lem:restriction}. On the other hand one can prove that 
$i^*(\bb{L}_2)=\bb{L}_2|_{S_0\sqcup S_1}$ is the trivial element, where
$\bb{L}_k\simeq{\bb{L}_1}^{\otimes k}$ (see the proof of Proposition \ref{ea:sesq} below). From that one infers  that the kernel of the homomorphism $i^* : H^2_{\Z_2}(L_{2q},\Z(1)) \to H^2_{\Z_2}(L_{2q}^\tau,\Z(1))$ is $\Z_q \subset \Z_{2q}$. Hence the exact sequence above reduces to the short exact sequence
\begin{equation}\label{ea:sesq}
0 \to \Z_2 \to H^2_{\Z_2}(L_{2q}| L_{2q}^\tau,\Z(1)) \to \Z_q \overset{i^*}{\to} 0\;.
\end{equation}
\begin{proposition} It holds true that
\[
H^2_{\Z_2}(L_{2q}| L_{2q}^\tau,\Z(1)) \;\simeq\; \Z_{2q}\;.
\]
\end{proposition}
\begin{proof}
The result follows if one can show that the exact sequence \eqref{ea:sesq} is not splitting. Let us start by proving that 
$\bb{L}_2|_{S_0\sqcup S_1}$ admits a nowhere vanishing ``Real'' section and therefore is trivial. For that it is suficient to consider the two sections $\sigma_j:S_l\to \bb{L}_2|_{S_j}$, with $j=0,1$, defined by
\begin{align*}
\sigma_0\left(\left[\cos \frac{\theta}{2}, \sin \frac{\theta}{2}\right]\right) 
&\;=\; \left[\left(\cos \frac{\theta}{2}, \sin \frac{\theta}{2}\right), 1\right]\;, \\
\sigma_1\left(\left[\expo{-\ii\frac{\pi}{2q}}\cos \frac{\theta}{2}, 
\expo{-\ii\frac{\pi}{2q}}\sin \frac{\theta}{2}\right]\right) 
&\;=\;
\left[\left(\expo{-\ii\frac{\pi}{2q}}\cos \frac{\theta}{2}, 
\expo{-\ii\frac{\pi}{2q}}\sin \frac{\theta}{2}\right),\expo{\ii\frac{\pi}{q}}\right] \;,
\end{align*}
for every $\theta\in \R/2\pi \Z$. 
The pair $(\mathscr{L}_2, \sigma_0 \sqcup \sigma_1)$ represents an element of $H^2_{\Z_2}(L_{2q}| L_{2q}^\tau,\Z(1))$ which surjects to a basis of the kernel $\Z_q$ of $i^*$. Therefore, the short exact sequence \ref{ea:sesq} would be split if and only if $(\mathscr{L}_2, \sigma_0 \sqcup \sigma_1)^{\otimes q}$ were trivial. We readily see
\[
(\mathscr{L}_{2}, \sigma_0 \sqcup \sigma_1)^{\otimes q}
\;\simeq\; (\mathscr{L}_{0}, 1 \sqcup (-1))\,
\]
where $\mathscr{L}_{0} \cong \mathscr{L}_1^{\otimes2q}$ is the trivial line bundle. Since $\sigma_1^{\otimes q} = -1$ is not trivial, so is the element above.
\end{proof}

\medskip
The latter result provides the correct version of  \cite[Proposition 5.4]{denittis-gomi-18-II}, which is consistent with Theorem \ref{thm:main_in_the_body}.

\subsection{Physical applications}
Before concluding, let us suggest some physical situation where the classification of the “Quaternionic” structures over $L_{2q}$ 
described in Section \ref{sec:Q_str_len} can be of some relevance.

\medskip

First of all, it is worth mentioning that the lens spaces $L_{n}$  naturally enter the  theory of the \emph{Dirac monopole}. In fact, it can be shown that $L_{n}$ is isomorphic to the  line bundle $\rr{h}^{\otimes n}$, where $\rr{h}$ is the dual bundle of the tautological line bundle over $\n{CP}^1$, and the magnetic monopole of charge $n\neq 0$ is the curvature of the rotationally invariant $\n{U}(1)$ connection over $L_{n}$
\cite[Appendix A]{jante-schroers-14}.

\medskip

The Born-Oppenheimer approximation permits to construct interesting examples of Topological Quantum System (with symmetries) in the sense described in Section \ref{sect:intro} \cite{baer-06,faure-zhilinskii-01,gat-robbins-15}.
The Born-Oppenheimer approximation can generally be applied when a quantum system is coupled with another comparatively slower system which is treated classically. In quantum mechanics this occurs, for example, in molecular dynamics, where usually the electrons have a fast motion compared to the motion of nuclei.
Let us denote with $X$ the \emph{classical state space} (phase space).
For a fixed classical state $x\in X$, one considers an instantaneous operator $H(x)$ acting on the \emph{quantum state space} 
 $\s{H}\simeq\C^M$ which describes the dynamics of the fast degrees of freedom.   One  immediately recognizes that this framework is summarized by \eqref{eq:intro_tqs0}, which provides the definition of a TQS. 
 
 \medskip
 
 Now, being more specific, let us assume that $X$ is the classical state space of a particle of mass $m$ constrained on the unit sphere $\n{S}^2\subset\R^3$. Therefore,  the position of the particle is specified by a vector $q\in\R^3$ such that $|q|=1$. The momentum of the particle is $p=mv$ where $v$ is the velocity. Since the velocity is tangent to the sphere one gets that $q\cdot p=m(q\cdot v)=0$. If we denote with $|p|$ the modulus of the momentum and with $\wp:=p/|p|$ its unit vector one obtains that 
\[
X\;=\;\Omega\;\times\; [0,+\infty)
\] 
 where
 \[
 \Omega\;:=\;\left\{(q,\wp)\in\n{S}^2\times \n{S}^2\;|\; q\cdot \wp=0\right\}\;.
 \]
 Since $[0,+\infty)$ is contractible, $\Omega$ provides the only relevant part for the analysis of topological effects. Following the construction in \cite[Section III]{bharath-18} one can prove that  $\Omega\simeq L_2$.
 Let $i:\n{S}^3\to S\n{U}(2)$ be the standard identification given by
 \[
 i\;:\;(z_0,z_1)\;\longmapsto\;
 \left(\begin{array}{cc}
 z_0 & z_1 \\
-\overline{z_1} & \overline{z_0}
 \end{array}\right)\;=\;\expo{\ii d \cdot \sigma \frac{\theta}{2}}\;,
 \]
 where $\sigma=(\sigma_1, \sigma_2, \sigma_3)$ is the vector of the Pauli matrices, $\theta=2\arccos({\rm Re}(z_0))$
and the unit vector $d$ is given by
\[
d\:=\:\frac{1}{\sin  \frac{\theta}{2}}\left({\rm Im}(z_0),{\rm Im}(z_1),{\rm Re}(z_1)\right)\;.
\] 
 Let
 $f:S\n{U}(2)\to S\n{O}(3)$ be the standard double cover given by
 \[
 f\;:\; \expo{\ii d \cdot \sigma \frac{\theta}{2}}\;\longmapsto\;R_d(\theta)
 \]
 where $R_d(\theta)\in S\n{O}(3)$ is the matrix that rotates of an angle $\theta$ around the direction $d$. 
 Observe that $f^{-1}(R_d(\theta))\mapsto \{\pm \expo{\ii d \cdot \sigma \frac{\theta}{2}}\}$.
 Finally, for a fixed $(q_0,\wp_0)\in  \Omega$ let $g:S\n{O}(3)\to \Omega$ be the map given by
\begin{equation}\label{eq:map_g}
g\;:\;R_d(\theta)\;\longmapsto\;R_d(\theta)\cdot(q_0,\wp_0)\;:=\;\left(R_d(\theta)q_0,R_d(\theta)\wp_0\right)\;.
\end{equation}
One can  directly check that 
the map $\alpha:\n{S}^3\to \Omega$, defined by
$\alpha:=g\circ f\circ i$ is a double covering map in view of the fact that $f$ is a double cover. In particular, if $(q,\wp)= R_d(\theta)\cdot(q_0,\wp_0)$, then
\[
\alpha^{-1}\;:\;(q,\wp)\;\longmapsto\;\{(z_0,z_1),(-z_0,-z_1))\}\;\in\;L_2
\]
 provides the desired identification $\Omega\simeq L_2$. Let us identify 
now the action of the involution $\tau':=\alpha\circ \tau$ on $\Omega$ induced by the involution $\tau$ on $L_2$. Given the unit vector $d=(d_1,d_2,d_3)$,  let $\overline{d}:=(-d_1,-d_2,d_3)$. Then a simple check shows that
$$
\tau'\;:\; R_d(\theta)\cdot(q_0,\wp_0)\;\longmapsto\;R_{\overline{d}}(\theta)\cdot(q_0,\wp_0)\;.
$$
The pair $(\Omega, \tau')$ turns out to be an involutive space equivalent to $(L_2,\tau)$.

 \medskip
 
Let assume now that the operator $H$ acting on the  fast degrees of freedom depends in the momentum $p$  only through   a term of the type $|p\cdot w|\propto|\wp\cdot w|$ for a fixed unitary vector $w\in\R^3$. In this case the relevant classical degree of freedom for the momentum is the line detected by $\wp$, and denoted by $\ell_\wp$, rather than $\wp$ itself. Therefore the relevant classical state space becomes
 \[
 \Sigma\;:=\;\left\{(q,\ell_\wp)\in\n{S}^2\times \n{RP}^3\;|\; q\cdot \wp=0\right\}\;
 \]
 which is the  space of all tangent lines to the sphere.
A generalization of the argument above shows that $\Sigma\simeq L_4$ \cite[Section III]{bharath-18}. The main difference now consists in the fact that the map $g$, defined as in \eqref{eq:map_g},
becomes a double covering. In fact, given a reference point $(q_0,\ell_{\wp_0})$ and a generic point $(q,\ell_\wp)=R_d(\theta)\cdot(q_0,\ell_{\wp_0})$, it holds true that
\[
g^{-1}\;:\;(q,\ell_\wp)\;\longmapsto\;\left\{R_d(\theta), R_d(\theta)R_q(\pi)\right\}
\]
in view of the relations $R_q(\pi)q=q$, $R_q(\pi)\wp=-\wp$ and $R_q(\pi)\ell_\wp=\ell_{-\wp}=\ell_\wp$. As a consequence, the map $\alpha$  
turns out to be a 4-covering as a composition of two double covering
and $\alpha^{-1}$ provides the identification $\Sigma\simeq L_4$.

\appendix

\section{A short reminder of the equivariant Borel cohomology}
\label{subsec:borel_cohom}
The proper cohomology theory for the analysis of vector bundle theories in the category of spaces with involution is the {equivariant cohomolgy} introduced by  A.~Borel in \cite{borel-60}. This cohomology has been used for the topological classification of \virg{Real} vector bundles \cite{denittis-gomi-14} and plays also a role in the classification of \virg{Quaternionic} vector bundles \cite{denittis-gomi-14-gen,denittis-gomi-18-I,denittis-gomi-18-II}. A short   self-consistent summary of this cohomology theory can be found in \cite[Section 5.1]{denittis-gomi-14} and we refer to \cite[Chapter 3]{hsiang-75} and \cite[Chapter 1]{allday-puppe-93}
for a more detailed introduction to the subject.

\medskip

Let us briefly recall the main steps of the
\emph{Borel construction}. 
The \emph{homotopy quotient} of an involutive space   $(X,\tau)$ is the orbit space
\begin{equation}\label{eq:homot_quot}
{X}_{\sim\tau}\;:=\;X\times\ {\n{S}}^{0,\infty} /( \tau\times \theta_\infty)\;.
\end{equation}
Here $\theta_\infty$ is the {antipodal map} on the infinite sphere $\n{S}^\infty$ 
(\cf \cite[Example 4.1]{denittis-gomi-14}) and ${\n{S}}^{0,\infty}$ is used as short notation for the pair $(\n{S}^\infty,\theta_\infty)$.
The product space $X\times{\n{S}}^\infty$ (forgetting for a moment the $\Z_2$-action) has the \emph{same} homotopy type of $X$ 
since $\n{S}^\infty$ is contractible. Moreover, since $\theta_\infty$ is a free involution,  also the composed involution $\tau\times\theta_\infty$ is free, independently of $\tau$.
Let $\s{R}$ be any commutative ring (\eg, $\R,\Z,\Z_2,\ldots$). The \emph{equivariant cohomology} ring 
of $(X,\tau)$
with coefficients
in $\s{R}$ is defined as
$$
H^\bullet_{\Z_2}(X,\s{R})\;:=\; H^\bullet({X}_{\sim\tau},\s{R})\;.
$$
More precisely, each equivariant cohomology group $H^j_{\Z_2}(X,\s{R})$ is given by the
 singular cohomology group  $H^j({X}_{\sim\tau},\s{R})$ of the  homotopy quotient ${X}_{\sim\tau}$ with coefficients in $\s{R}$ and the ring structure is given, as usual, by the {cup product}.
As the coefficients of
the usual singular cohomology are generalized to \emph{local coefficients} (see \eg \cite[Section 3.H]{hatcher-02} or
\cite[Section 5]{davis-kirk-01}), the coefficients of the Borel  equivariant cohomology are also
generalized to local coefficients. Given an involutive space $(X,\tau)$ let us consider the homotopy group $\pi_1({X}_{\sim\tau})$
and the associated  \emph{group ring} $\Z[\pi_1({X}_{\sim\tau})]$. Each module $\s{Z}$ over the group $\Z[\pi_1({X}_{\sim\tau})]$ is, by definition,
a \emph{local system} on $X_{\sim\tau}$.  Using this local system one defines, as usual, the equivariant cohomology with local coefficients in $\s{Z}$:
$$
H^\bullet_{\Z_2}(X,\s{Z})\;:=\; H^\bullet({X}_{\sim\tau},\s{Z})\;.
$$
We are particularly interested in modules $\s{Z}$ whose underlying groups are identifiable with $\Z$. 
For each involutive space  $(X,\tau)$, there always exists a particular family of local systems $\Z(m)$
labelled by $m\in\Z$. Here
 $\Z(m)\simeq X\times\Z$ denotes the $\Z_2$-equivariant local system on $(X,\tau)$  made equivariant  by the $\Z_2$-action $(x,l)\mapsto(\tau(x),(-1)^ml)$.
Because the module structure depends only on the parity of $m$, we consider only the $\Z_2$-modules ${\Z}(0)$ and ${\Z}(1)$. Since ${\Z}(0)$ corresponds to the case of the trivial action of $\pi_1(X_{\sim\tau})$ on $\Z$ one has $H^k_{\Z_2}(X,\Z(0))\simeq H^k_{\Z_2}(X,\Z)$ \cite[Section 5.2]{davis-kirk-01}.

\medskip

We recall the two important group isomorphisms 
\begin{equation}\label{eq:iso:eq_cohom}
\begin{aligned}
H^1_{\Z_2}\big(X,\Z(1)\big)\;&\simeq\;\big[X,\n{U}(1)\big]_{\Z_2}\;,\\ H^2_{\Z_2}\big(X,\Z(1)\big)\;&\simeq\;{\rm Vec}_{\rr{R}}^1\big(X,\tau\big)\equiv {\rm Pic}_{\rr{R}}\big(X,\tau\big)\;,
\end{aligned}
\end{equation}
 involving the 
first two equivariant cohomology groups. 
The first isomorphism \cite[Proposition A.2]{gomi-15} says that the first equivariant cohomology group is isomorphic to the set of $\Z_2$-equivariant homotopy classes of $\Z_2$-equivariant maps $\varphi:X\to\n{U}(1)$ where the involution on $\n{U}(1)$ is induced by the complex conjugation, \ie $\varphi(\tau(x))=\overline{\varphi(x)}$. The second isomorphism is due to B.~Kahn \cite{kahn-59} and 
expresses the equivalence between the Picard group of \virg{Real} line bundles (in the sense of \cite{atiyah-66,denittis-gomi-14}) over  $(X,\tau)$ and the second equivariant cohomology group of this space.

\medskip

The fixed point subset $X^\tau\subset X$ is closed and $\tau$-invariant and the inclusion $\imath:X^\tau\hookrightarrow X$ extends to an inclusion $\imath:X^\tau_{\sim\tau}\hookrightarrow X_{\sim\tau}$ of the respective homotopy quotients. The \emph{relative} equivariant cohomology can be defined as usual by the identification
$$
H^\bullet_{\Z_2}\big(X|X^\tau,\s{Z}\big)\;:=\; H^\bullet\big({X}_{\sim\tau}|X^\tau_{\sim\tau},\s{Z}\big)\;.
$$
Consequently, one has the related long exact sequence in cohomology
\begin{equation}\label{eq:Long_seq}
\ldots\;H^k_{\Z_2}\big(X|X^\tau,\s{Z}\big)\;\stackrel{}{\to}\;H^k_{\Z_2}\big(X,\s{Z}\big)\;\stackrel{r}{\to}\;H^k_{\Z_2}\big(X^\tau,\s{Z}\big)\;\stackrel{}{\to}\;H^{k+1}_{\Z_2}\big(X|X^\tau,\s{Z}\big)\;\ldots
\end{equation}
where the map $r:=\imath^*$ restricts  cochains on $X$ to  cochains on $X^\tau$. The $k$-th \emph{cokernel} of $r$ is by definition
$$
{\rm Coker}^k\big(X|X^\tau,\s{Z}\big)\;:=\;H^k_{\Z_2}\big(X^\tau,\s{Z}\big)\;/\;r\big(H^k_{\Z_2}(X,\s{Z})\big)\;.
$$

\medskip

Let us point out that with the same construction 
one can define {relative} cohomology theories $H^\bullet_{\Z_2}(X|Y,\s{Z})$
for each $\tau$-invariant subset $Y\subset X^\tau$, or more in general for every  $\Z_2$-CW pair $(X, Y)$ consisting of a  $\Z_2$-CW complex $X$, a $\Z_2$-CW subcomplex $Y$ and a 
sub-complex inclusion $Y\hookrightarrow X$ \cite[Remark 1.2.10]{allday-puppe-93}.
 If $Y=\emptyset$ then $H^k_{\Z_2}(X|\emptyset,\s{Z})\simeq H^k_{\Z_2}(X,\s{Z})$ by definition, hence it is reasonable to put $H^k_{\Z_2}(\emptyset,\s{Z})=0$ for consistency with the above long exact sequence.
The case $Y:=\{\ast\}$ of a single invariant point is important since it defines the \emph{reduced} cohomology theory  
$$
\widetilde{H}^k_{\Z_2}\big(X,\s{Z}\big)\;:=\;H^k_{\Z_2}\big(X|\{\ast\},\s{Z}\big)\;.
$$
In this case, the obvious surjectivity of the map $r$ at each step of the exact sequence \eqref{eq:Long_seq} justifies the isomorphism
\begin{equation}\label{eq:app_red_cohom1}
H^k_{\Z_2}\big(X,\s{Z}\big)\;\simeq\;\widetilde{H}^k_{\Z_2}\big(X,\s{Z}\big)\;\oplus {H}^k_{\Z_2}\big(\{\ast\},\s{Z}\big)
\end{equation}

\section{Cohomology of the Eilenberg-Mac Lane space}
\label{subsec:cohom_EilenbergMacLane}
Let us recall the notation
$$
\s{K}\;:=\;K(0 \rightsquigarrow \tilde{\Z}, 2)
$$ 
introduced in Section
\ref{sec:cohom_eilMC}. {The homotopy type of the spaces $\s{K}$ and $\s{K}^\tau$ is described in the following result.
\begin{lemma}\label{lemm:homot_typ}
It holds true that
\begin{equation}\label{eq_homot}
{\pi}_k(\s{K})\;=\;
\left\{
\begin{array}{ll}
\Z & (k=2) \\
0 & (k \neq n)\;
\end{array}
\right.
\end{equation}
and 
\begin{equation}\label{eq_homot00}
{\pi}_k(\s{K}^\tau)\;=\;0\;,\qquad \forall\; k\geqslant 0\;.
\end{equation}
As a consequence, forgetting  the $\Z_2$-action, $\s{K}$ is homotopy equivalent  to the classical Eilenberg-Mac Lane space $K(\Z, 2) \simeq \n{CP}^\infty$. Similarly, $\s{K}^\tau$ is  homotopy equivalent  to a singleton $\{\ast\}$.
\end{lemma}
\begin{proof}
Using the notation of Example \ref{ex:orb_catZ2} one gets $\s{K}^{\n{H}_0}=\s{K}$ and $\s{K}^{\n{H}_1}=\s{K}^\tau$. Therefore, in view of the definition in Example \ref{ex-homot} one obtains
\[
\underline{\pi}_k[\s{K}](Z_0)\;=\;\pi_k(\s{K})\;,\qquad \underline{\pi}_k[\s{K}](Z_1)\;=\;\pi_k(\s{K}^\tau)\;.
\]
Comparing these latter equations with the defining property 
\[
\underline{\pi}_k(\s{K})\;=\;
\left\{
\begin{array}{ll}
0\rightsquigarrow \widetilde{\Z} & (k=2) \\
0 & (k \neq 2)\;,
\end{array}
\right.
\]
of the 
Eilenberg-Mac Lane space $\s{K}$, one gets equations \eqref{eq_homot} and \eqref{eq_homot00}.
Now, in view of  the uniqueness of the homotopy type of the classical
Eilenberg-MacLane spaces
\cite[Proposition 4.30]{hatcher-02} there is a weak homotopy equivalence between $\s{K}$ and $K(\Z, 2)$, and between $\s{K}^\tau$ and the singleton. Finally, the Whitehead’s Theorem \cite[Theorem 4.5]{hatcher-02} ensures that the weak homotopy equivalences above induce respective
homotopy equivalences. \end{proof}
}

\begin{lemma}\label{lemma_tool1}
There are isomorphisms
\begin{align*}
H^n(\s{K},\Z)\;
\simeq\; \mathrm{Hom}(\pi_n(\s{K}), \Z)\;,\qquad n=1,2\;.
\end{align*}
Moreover, the first integral cohomology groups of  $\s{K}$ are 
$$
\begin{array}{|c||c|c|c|c|c|}
\hline
& n = 0 & n = 1 & n = 2 & n = 3 & n = 4 \\
\hline
\hline
H^n(\s{K}, \Z) & \Z & 0 & \Z & 0 & \Z \\
\hline
\end{array}
$$
\end{lemma}
\begin{proof}
Let us use the homotopy equivalence $\s{K}\simeq K(\Z, 2) \simeq \n{CP}^\infty$ from Lemma \ref{lemm:homot_typ}.
Since integral cohomology of $\n{CP}^\infty$ is well-known, one has that
 $H^j(\s{K}, \Z)\simeq  \Z$ for $j$ even and $H^j(\s{K}, \Z) \simeq 0$ for $j$ odd. Moreover, the Hurewicz homomorphism $\pi_2(\s{K}) \to H_2(\s{K})$ is an isomorphism. By the universal coefficient theorem, the homomorphism $H^n(\s{K},\Z) \to \mathrm{Hom}(H_n(\s{K}), \Z)$ is an isomorphism for $0\leqslant n \leqslant 2$. 
\end{proof}

\begin{lemma}\label{lemma_tool2}
The Borel equivariant cohomology groups  in low degrees of the space $\s{K}$ are summarized in the following table:
$$
\begin{array}{|c||c|c|c|c|}
\hline
& n = 0 & n = 1 & n = 2 & n = 3 \\
\hline
\hline
H^n_{\Z_2}(\s{K},\Z) & \Z & 0 & \Z_2 & \Z_2 \\
\hline
H^n(\s{K},\Z) & \Z & 0 & \Z & 0 \\
\hline
H^n_{\Z_2}(\s{K},\Z(1)) & 0 & \Z_2 & \Z & \Z_2 \\
\hline
\end{array}
$$

\noindent
In particular,   the map  that forgets the $\Z_2$-action
provides the isomorphism 
$$ 
H^2_{\Z_2}(\s{K},\Z(1)) \;\overset{f}{\simeq}\; H^2(\s{K},\Z)\;\simeq\; \mathrm{Hom}(\pi_2(\s{K}), \Z)\;.
$$
\end{lemma}
\begin{proof}
The strategy of the proof is very similar to that of \cite[Lemma A.2]{denittis-gomi-18-I} that can be used as a reference for more details.
To compute the Borel equivariant cohomology, we use the spectral sequence
\[
E_2^{p, q} \;=\; H^p(\Z_2, H^q(\s{K},\Z))\; \Rightarrow\;
H^\bullet_{\Z_2}(\s{K},\Z)\;,
\]
where the coefficient $H^q(\s{K},\Z)$ in the group cohomology of $\Z_2$ is endowed by the $\Z_2$-action induced from the involution on $\s{K}$. As it has been seen in Lemma \ref{lemma_tool1} one has
\[
H^q(\s{K},\Z)\; \simeq\; \mathrm{Hom}(\pi_q(\s{K}), \Z)\;,\qquad q=1,2\;.
\]
If one takes the $\Z_2$-action into account, then $H^0(\s{K},\Z) \simeq \Z$ and $H^2(\s{K},\Z) \simeq \widetilde{\Z}$ by the very definition of the Eilenberg-Mac Lane space $\s{K}$. Hence the $E_2$-terms can be summarized as follows:
$$
\begin{array}{c||c|c|c|c|c}
q = 3 & 0 & 0 & 0 & 0 & 0 \\
\hline
q = 2 & 0 & \Z_2 & 0 & \Z_2 & 0 \\
\hline
q = 1 & 0 & 0 & 0 & 0 & 0 \\
\hline
q = 0 & \Z & 0 & \Z_2 & 0 & \Z_2 \\
\hline
\hline
E_2^{p, q} & p = 0 & p = 1 & p = 2 & p = 3 & p = 4
\end{array}
$$

\noindent
This immediately determines $H^n_{\Z_2}(\s{K},\Z)$ for $n=0,1,2$. Note that the Eilenberg-Mac Lane space is assumed to have a fixed point $\ast\in\s{K}^\tau\neq\emptyset$. Then $E_2^{p, 0}$ must survive into the direct summand $H^p_{\Z_2}(\{\ast\},\Z)$ in the decomposition 
$$
H^p_{\Z_2}(\s{K},\Z)\; \simeq\; H^p_{\Z_2}(\{\ast\},\Z)\; \oplus\; \widetilde{H}^p_{\Z_2}(\s{K},\Z)
$$ 
by using the reduced cohomology.
This shows that 
$$
\Z_2\;\simeq\;E_2^{1, 2}\; =\; E_3^{1, 2}\;=\;\ldots\; =\; E_\infty^{1, 2} \;\simeq\; H^3_{\Z_2}(\s{K},\Z)\;.
$$
Let us now  make use the spectral sequence
$$
E_2^{p, q} \;=\; H^p(\Z_2, H^q(\s{K},\Z) \otimes \widetilde{\Z})
\;\Rightarrow\; H^\bullet_{\Z_2}(\s{K},\Z(1))\;.
$$
The $E_2$-terms are summarized as follows:
$$
\begin{array}{c||c|c|c|c|c}
q = 3 & 0 & 0 & 0 & 0 & 0 \\
\hline
q = 2 & \Z & 0 & \Z_2 & 0 & \Z_2 \\
\hline
q = 1 & 0 & 0 & 0 & 0 & 0 \\
\hline
q = 0 & 0 & \Z_2 & 0 & \Z_2 & 0 \\
\hline
\hline
E_2^{p, q} & p = 0 & p = 1 & p = 2 & p = 3 & p = 4
\end{array}
$$

\noindent
This also immediately determines $H^n_{\Z_2}(\s{K},\Z(1))$ for $n=0,1,2$. By the same argument about a fixed point $\ast\in\s{K}^\tau$ and the reduced cohomology, one also see that 
$$
\Z_2\;\simeq\;E_2^{3, 0}\; =\; E_3^{3, 0}\; =\;\ldots\; =\;E_\infty^{3, 0} \;\simeq\; H^3_{\Z_2}(\s{K},\Z(1))\;.
$$
From the exact sequence  in \cite[Proposition 2.3]{gomi-15}
$$
\underbrace{H^1_{\Z_2}(\s{K},\Z)}_{0} \to
\underbrace{H^2_{\Z_2}(K,\Z(1))}_{\Z} \overset{f}{\to}
\underbrace{H^2(\s{K},\Z)}_{\Z} \to
\underbrace{H^2_{\Z_2}(\s{K},\Z)}_{\Z_2} \to
\underbrace{H^3_{\Z_2}(K,\Z(1)}_{\Z_2} \overset{}{\to}
\underbrace{H^3(\s{K},\Z)}_0 
$$
one infers  that $f $ is an isomorphism.
\end{proof}

\medskip

Recall that 
the inclusion $j:(\s{K}, \emptyset) \hookrightarrow (\s{K}, \s{K}^\tau)$ induces a natural homomorphism $j^*$ in the Borel equivariant cohomology.

\begin{lemma}\label{lemma_tool3}
The inclusion  $j$ induces the isomorphism
$$ 
H^2_{\Z_2}(\s{K}|\s{K}^\tau,\Z(1)) \;\overset{j^*}{\simeq}\; H^2_{\Z_2}(\s{K},\Z(1))\;.
$$
\end{lemma}
\begin{proof}
By  Lemma \ref{lemm:homot_typ}  $\s{K}^\tau \neq \emptyset$ is homotopy equivalent to a singleton $\{\ast\}$, on which evidently $\Z_2$ acts trivially. Therefore one has that
$$
H^{n}_{\Z_2}(\s{K}^\tau,\Z(1))\;\simeq\; H^{n}_{\Z_2}(\{\ast\},\Z(1))
$$
and
$$
H^{n}_{\Z_2}(\s{K}|\s{K}^\tau,\Z(1))\;\simeq\; \widetilde{H}^{n}_{\Z_2}(\s{K},\Z(1))
$$

\noindent
where on the right hand side there is the reduced cohomology induced by the inclusion $\{\ast\} \hookrightarrow \s{K}$. The Borel equivariant cohomology of the fixed point is well known \cite[Proposition 2.4.]{gomi-15} and in particular $H^{2}_{\Z_2}(\{\ast\},\Z(1))=0$. From \eqref{eq:app_red_cohom1} one gets
$$
H^{2}_{\Z_2}(\s{K},\Z(1))\;\simeq\; \widetilde{H}^{2}_{\Z_2}(\s{K},\Z(1))\;
\simeq\;H^{2}_{\Z_2}(\s{K}|\s{K}^\tau,\Z(1))
$$

\noindent
where the isomorphism is induced by the 
inclusion.
\end{proof}



\begin{thebibliography} {CTSR}
\frenchspacing 






\bibitem[AB]{aharonov-bohm-59} 
 Aharonov, Y.; Bohm, D.: {\sl Significance of electromagnetic potentials in quantum theory}.  Phy. Rev. {\bf 115}, 485-491, (1959)
 
\bibitem[AF]{ando-fu-15} {Ando, Y.;~Fu, L.}: 
{\sl  Topological crystalline insulators and topological superconductors: from concepts to materials}. Annu. Rev. Cond. Matt. Phys. {\bf 6}, 361-381 (2015) 


\bibitem[AM]{ashcroft-mermin-76} {Ashcroft, N. W.;   Mermin N. D.}: 
{\em Solid State Physics}. 
Saunders College Pub., Philadelphia, 1976

\bibitem[AP]{allday-puppe-93} {Allday, C.;~Puppe, V.}: 
{\em Cohomological Methods in Transformation Groups}. 
Cambridge University Press, Cambridge, 1993



 \bibitem[Ati1]{atiyah-66}
{Atiyah, M.~F.}: {\sl $K$-theory and reality}. 
Quart. J. Math. Oxford Ser. (2) {\bf 17}, 367-386 (1966)

\bibitem[Ati2]{atiyah-67}
{Atiyah, M.~F.}: {\sl $K$-theory}. 
W. A. Benjamin, New York, 1967



\bibitem[Bae]{baer-06} {Baer, M.}: {\em Beyond Born-Oppenheimer: Electronic Nonadiabatic Coupling Terms and Conical Intersections}. 
Wiley \& Sons,  Hobokenn, 2006


\bibitem[Ber]{berry-84}
  Berry M. V.: {\sl Quantal Phase Factors Accompanying Adiabatic Changes}. Proc. Roy. Soc. Lond. A. {\bf 392}, 45-57, (1984)

\bibitem[BES]{bellissard-elst-schulz-baldes-94} Bellissard,~J.; van~Elst,~A.; Schulz-Baldes,~H.: {\sl The Non-Commutative Geometry of the Quantum Hall Effect}. 
{J. Math. Phys.}~{\bf 35}, 5373-5451 (1994)

\bibitem[Bha]{bharath-18} 
{Bharath H. M.:}
{\sl Non-Abelian geometric phases carried by the spin fluctuation tensor}
J. Math. Phys. {\bf 59}, 062105 (2018)



\bibitem[BMKNZ]{bohm-mostafazadeh-koizumi-niu-zwanziger-03} {B\"{o}hm, A.; and Mostafazadeh, A.;  Koizumi, H.;  Niu, Q.;  Zwanziger, J.}: {\em The Geometric Phase in Quantum Systems}. 
Springer-Verlag, Berlin, 2003





\bibitem[Bor]{borel-60} Borel,~A.: {\em Seminar on transformation groups}, with contributions by G.~Bredon, E.~E.~Floyd, D.~Montgomery, R.~Palais. Annals of Mathematics Studies {\bf 46}, Princeton University Press, Princeton,  1960


\bibitem[Bre]{Bre} {Bredon, G.~E.}: 
{\em Equivariant cohomology theories}.
Lecture Notes in Mathematics  No. {\bf 34}, Springer-Verlag, Berlin-New York, 1967

\bibitem[BT]{bott-tu-82} {Bott, R.;  Tu, L. W.}: {\em Differential Forms in Algebraic Topology}. 
Springer-Verlag, Berlin, 1982




\bibitem[CJ]{chruscinski-jamiolkowski-04} {Chru\'{s}ci\'{n}ski, D.; Jamio{\l}kowski, A.}: {\em Geometric Phases in Classical and Quantum Mechanics}. 
Birkh\"{a}user, Basel, 2004




\bibitem[DG1]{denittis-gomi-14} De~Nittis,~G.; Gomi,~K.: 
{\sl 
Classification of \virg{Real} Bloch-bundles: Topological Insulators of type AI}. J. Geometry Phys.~{\bf 86}, 303-338 (2014)


\bibitem[DG2]{denittis-gomi-14-gen} De~Nittis,~G.; Gomi,~K.: 
{\sl 
Classification of \virg{Quaternionic} Bloch-bundles: Topological Insulators of type AII}. Commun. Math. Phys.~{\bf 339}, 1-55 (2015)

\bibitem[DG3]{denittis-gomi-18-I} De~Nittis,~G.; Gomi,~K.: 
{\sl The cohomological nature of the Fu-Kane–Mele invariant}. 
J. Geometry Phys. ~{\bf 124}, 124-164 (2018)

\bibitem[DG4]{denittis-gomi-18-II} De~Nittis,~G.; Gomi,~K.: 
{\sl The FKMM-invariant in low dimension}. 
Lett. Math. Phys.~{\bf 108}, 1225-1277 (2018)

\bibitem[DG5]{denittis-gomi-18-III} De~Nittis,~G.; Gomi,~K.: 
{\sl Chiral vector bundles}. 
Math. Z.~{\bf 290}, 775-830 (2018)

\bibitem[DG6]{denittis-gomi-22} De~Nittis,~G.; Gomi,~K.: 
{\sl The Cohomology Invariant for Class DIII Topological Insulators}. 
Ann. Henri Poincar\'e~{\bf 23}, 775-830 (2022)


\bibitem[Dir]{dirac-31}
Dirac, P. A. M.:   {\sl Quantized singularities in the electromagnetic  field}.  
Proc. Roy. Soc. Lond. A. {\bf 133}, 60-72 (1931)     

\bibitem[DK]{davis-kirk-01} Davis, J. F.; Kirk, P.: {\em Lecture Notes in Algebraic Topology}. 
AMS, Providence, 2001




\bibitem[Dup]{dupont-69} Dupont,~J. L.: {\sl
Symplectic Bundles and $KR$-Theory}. 
{Math. Scand.}~{\bf 24}, 27-30 (1969)






\bibitem[ES]{eilenberg-steenrod-52} 
{Eilenberg, S.; Steenrod, N. E.}: 
{\em  Foundations of algebraic topology}.
 Princeton 1952 





\bibitem[FKM]{fu-kane-mele-95} 
{Fu, L.; Kane, C. L.; Mele, E. J.}: {\sl Topological Insulators in Three Dimensions}. {Phys. Rev. Lett.}~{\bf 98}, 106803 (2007)


\bibitem[Fre]{freyd-64} 
{Freyd, P.}: 
{\em Abelian Categories. An Introduction to the Theory of Functors.}.
Harper \& Row, 1964


\bibitem[FZ]{faure-zhilinskii-01} 
Faure,~F.; Zhilinskii,~B.: 
{\sl Topological Properties of the Born-Oppenheimer Approximation and Implications for the Exact Spectrum}. 
Lett. Math. Phys.~{\bf 55}, 219-238 (2001)




\bibitem[Gom]{gomi-15} 
{Gomi, K.}: {\sl A Variant of K-Theory and Topological T-Duality for Real Circle Bundles}. {Commun. Math. Phys.}~{\bf 334}, 923-975 (2015)


\bibitem[GR]{gat-robbins-15}  Gat, O.; Robbins, J. M.:  {\em Topology of time-invariant energy bands with adiabatic structure}.  
J. Phys. A: Math. Theor. {\bf 50}, (2017)





\bibitem[Hat]{hatcher-02} Hatcher,~A.: {\em  Algebraic Topology}. Cambridge University Press, Cambridge,  2002

\bibitem[HK]{hasan-kane-10} Hasan,~M.~Z.; Kane,~C.~L.: {\sl Colloquium: Topological insulators}. 
{Rev. Mod. Phys.}~{\bf 82}, 3045-3067 (2010)



\bibitem[Hsi]{hsiang-75} Hsiang, W. Y.: {\em  Cohomology Theory of Topological Transformation Groups}. Springer-Verlag, Berlin, 1975







\bibitem[JS]{jante-schroers-14} 
{Jante, R.; Schroers, B. J.:} 
{\sl Dirac operators on the Taub-NUT space, monopoles and SU(2) representations}. 
J. High Energ. Phys. {\bf 2014}, 114 (2014) 



\bibitem[Kah]{kahn-59} Kahn,~B.: 
{\sl Construction de classes de Chern \'{e}quivariantes pour un fibr\'{e} vectoriel R\'{e}el}. 
{Comm. Algebra.}~{\bf 15}, 695-711 (1987)


\bibitem[KM]{kane-mele-05} {Kane, C. L.; Mele, E. J.}: 
{\sl $\Z_2$ Topological Order and the Quantum Spin Hall Effect}. 
{Phys. Rev. Lett.}~{\bf 95}, 146802 (2005)

\bibitem[Kuc]{kuchment-93} Kuchment,~P.: {\em Floquet theory for partial differential equations}. Birkh\"{a}user, Boston,  1993




\bibitem[Ill]{illman-73} 
{Illman, S.}:
{\sl Equivariant singular homology and cohomology}. 
Bull. Amer. Math. Soc. {\bf 79}, 188-192(1973)  





\bibitem[May]{May} 
{May, J. P.}:
{\em Equivariant homotopy and cohomology theory}. 
With contributions by M. Cole, G. Comeza\~{n}a, S. Costenoble, 
A. D. Elmendorf, J. P. C. Greenlees, L. G. Lewis, Jr., 
R. J. Piacenza, G. Triantafillou, and S. Waner.
CBMS Regional Conference Series in Mathematics   {\bf 91}, 
AMS, Providence, 1996



\bibitem[MLa]{maclane-78} 
{Mac Lane, S.}:
{\em Categories for the Working Mathematician}. 
Springer, 1978



\bibitem[Pan]{pancharatnam-56} 
Pancharatnam S.: {\sl Generalized Theory of Interference, and Its Applications. Part I. Coherent Pencils}. 
Proc. Indian Acad. Sci. A. {\bf 44}, 247-262, (1956) 


\bibitem[Pet]{peterson-59} 
{Peterson, F. P.}: {\sl Some remarks on Chern classes}. {Ann. of Math.}~{\bf 69}, 414-420 (1959)



\bibitem[Se]{serre-55}
Serre, J.-P.:
{\sl Faisceaux Algebriques Coherents}.
Ann. Math. {\bf 61}, 197-278 (1955)



\bibitem[Spa]{spanier-66} Spanier, E. H.: {\em Algebraic Topology}. McGraw-Hill, New York, 1966

\bibitem[Sw]{swan-62}
Swan, R.-G.:
{\sl Vector Bundles and Projective Modules}.
Trans. Amer. Math. Soc.  {\bf 105}, 264?277 (1962)




\bibitem[TKNN]{thouless-kohmoto-nightingale-nijs-82} Thouless,~D.~J.; Kohmoto,~M.; Nightingale,~M.~P.; den
Nijs,~M.: {\sl Quantized Hall Conductance in a Two-Dimensional Periodic Potential}. 
{Phys. Rev. Lett.}~{\bf 49}, 405-408 (1982)






\bibitem[Yan]{yang-96}
Yang, C. N.: {\sl Magnetic Monopoles, Fiber Bundles, and Gauge Fields}.
In \em{History of Original Ideas and Basic Discoveries in Particle Physics} pg. 55-65, 
Springer, Boston, 1996





 \end{thebibliography}
\end{document}